\pdfoutput=1
\documentclass[a4paper,11pt,twoside]{article}

\usepackage[english]{babel}%

\usepackage[utf8]{inputenc}
\usepackage{cite}
\usepackage{etoolbox}
\usepackage{amsmath, amsthm}
\usepackage{amsfonts, amssymb}
\usepackage{bm}
\usepackage{ifthen}

\usepackage{letltxmacro}

\usepackage{nicefrac}
\usepackage{xstring}

\usepackage{pgf}
\usepackage{tikz}
\usetikzlibrary{arrows,automata}

\usepackage[T1]{fontenc}
\usepackage[tracking]{microtype}

\usepackage[ttscale=0.81]{libertine}
\usepackage[libertine]{newtxmath}

\renewcommand{\mathbf}[1]{\bm{#1}}

\usepackage{pdfsync, graphicx}

\usepackage{framed}

\usepackage{todo}

\usepackage{booktabs,siunitx}

\usepackage{caption}
\usepackage{subcaption}
\usepackage{wrapfig}

\usepackage[all]{nowidow}

\usepackage[bookmarks=false]{hyperref}
\hypersetup{pdfproducer={},
  pdfinfo={CreationDate={...},ModDate={...}},
  pdfauthor={},
  pdfkeywords={},
  pdfsubject={},
  pdfcreator={},
  pdflang={}
}
\hypersetup{
  colorlinks=true,
  linkcolor=green!30!black,
  linktoc=all,
  citecolor=blue
}

\usepackage{cleveref}

\newtheorem{theorem}{Theorem}[section]
\newtheorem{observation}[theorem]{Observation}
\newtheorem{claim}[theorem]{Claim}

\newtheorem{lemma}[theorem]{Lemma}
\newtheorem{corollary}[theorem]{Corollary}
\newtheorem{proposition}[theorem]{Proposition}
\newtheorem{consequence}[theorem]{Consequence}

\theoremstyle{definition}
\newtheorem{definition}{Definition}[section]

\crefname{theorem}{Theorem}{Theorems}
\crefname{observation}{Observation}{Observations}
\crefname{claim}{Claim}{Claims}
\crefname{condition}{Condition}{Conditions}
\crefname{example}{Example}{Examples}
\crefname{fact}{Fact}{Facts}
\crefname{lemma}{Lemma}{Lemmas}
\crefname{corollary}{Corollary}{Corollaries}
\crefname{consequence}{Consequence}{Consequences}
\crefname{definition}{Definition}{Definitions}
\crefname{remark}{Remark}{Remarks}
\crefname{proposition}{Proposition}{Propositions}
\newcommand{\abs}[1]{\ensuremath{\left|#1\right|}}

\newcommand{\ceil}[1]{\ensuremath{\left\lceil#1\right\rceil}}
\newcommand{\floor}[1]{\ensuremath{\left\lfloor#1\right\rfloor}}

\newcommand{\inb}[1]{\left\{#1\right\}}
\newcommand{\inp}[1]{\left(#1\right)}
\newcommand{\insq}[1]{\left[#1\right]}

\makeatletter
\newcommand*{\defeq}{\mathrel{\rlap{%
                     \raisebox{0.3ex}{$\m@th\cdot$}}%
                     \raisebox{-0.3ex}{$\m@th\cdot$}}%
                    =}
\newcommand*{\eqdef}{=
  \mathrel{\rlap{%
      \raisebox{0.3ex}{$\m@th\cdot$}}%
    \raisebox{-0.3ex}{$\m@th\cdot$}}%
}
\makeatother

\newcommand{\I}[1]{\ensuremath{\insq{#1}}}

\newcommand{\poly}[1]{\ensuremath{\mathop{\mathrm{poly}}\inp{#1}}}

\renewcommand{\vec}[1]{\mathbf{#1}}

\usetikzlibrary{arrows}

\usepackage{fullpage}

\usepackage{setspace}
\usetikzlibrary{er,positioning,bayesnet}
\usepackage{multicol}
\usepackage{verbatim}
\usepackage{algpseudocode,algorithm,algorithmicx}

\usepackage{enumerate}

\definecolor{blue}{HTML}{1F77B4}
\definecolor{orange}{HTML}{FF7F0E}
\definecolor{green}{HTML}{2CA02C}

\newcommand{\nckt}{neural circuit}
\newcommand{\nckts}{neural circuits}

\newcommand{\Nckts}{Neural circuits}
\newcommand{\parent}[1]{\ensuremath{p\inp{#1}}}
\newcommand{\nSimulate}{\textup{\textsc{Neural-Circuit-Simulation}}}

\newcommand{\ntv}[0]{non-trivial}
\newcommand{\Ntv}[0]{Non-trivial}

\newtheorem{problem}{Problem}

\begin{document}

\title{The PSPACE-hardness of understanding neural circuits}

\author{}
 \author{Vidya Sagar Sharma\thanks{Tata Institute of Fundamental Research,
     Mumbai. Email: \texttt{vidya.sagar@tifr.res.in}.} %
   \and Piyush Srivastava\thanks{Tata Institute of Fundamental Research,
     Mumbai. Email: \texttt{piyush.srivastava@tifr.res.in}.}}

 \date{}
\maketitle              %

\thispagestyle{empty}

\begin{abstract}
  In neuroscience, an important aspect of understanding the function of a
  \nckt{} is to determine which, if any, of the neurons in the circuit are
  \emph{vital} for the biological behavior governed by the \nckt{}: i.e., which
  sets of neurons, when deactivated, lead to an elimination of the behavior
  being studied.  Typically, one is interested in finding the smallest such
  sets.  A similar problem is to determine whether a given small set of neurons
  may be enough for the behavior to be displayed, even if all other neurons in
  the circuit are deactivated.  Such a subset of neurons form what is called a
  \emph{degenerate} circuit for the behavior being studied.

  Recent advances in experimental techniques have provided researchers with
  tools to activate and deactivate subsets of neurons with a very high
  resolution, even in living animals.  The data collected from such experiments
  may be of the following form: when a given subset of neurons is deactivated,
  is the behavior under study observed?

  This setting leads to the algorithmic question of determining the minimal
  vital or degenerate sets of neurons, when one is given as input a description
  of the \nckt{}.  The algorithmic problem entails both figuring out which
  subsets of neurons should be perturbed (activated/deactivated), and then using
  the data from those perturbations to determine the minimal vital or degenerate
  sets.  Given the large number of possible perturbations, and the recurrent
  nature of \nckts{}, the possibility of a combinatorial explosion in such an
  approach has been recognized in the biology and the neuroscience literature,
  e.g.~in a paper of Koch~(\emph{Science}, \textbf{337} (6094), pp.~531--532)
  and more recently in a paper of Kumar et al.~(\emph{Trends in Neurosciences}
  \textbf{36} (10), pp.~579--586).  In a recent paper, Ramaswamy
  (bior$\chi{}$iv, 2019) took a step towards formulating the question in terms
  of computational complexity theory and established NP-hardness for some of
  these problems.

  In this paper, we prove that the problems of finding minimal or minimum-size
  degenerate sets, and of finding the set of vital neurons, of a neural circuit
  given as input, are in fact PSPACE-hard.  Further, the hardness results hold
  even when all the neurons in the \nckt{} are threshold neurons with weights
  coming from a fixed, constant size set and have a bounded number of
  connections.  More importantly, we prove our hardness results by showing that
  a simpler problem, that of simulating such \nckts{}, is itself PSPACE-hard.

\end{abstract}

\newpage
\setcounter{page}{1}
\section{Introduction}
\label{sec:introduction}
\subsection{Background}
\label{sec:background}

In neuroscience, an important aspect of understanding the functioning of a
\nckt{} in a biological system is determining which neurons in the \nckt{} are
critical for the functioning of the system.  For example, Flood et
al.~\cite{flood_single_2013} showed that when the activity of a single specific
pair of neurons in \emph{Drosophila} is suppressed, certain feeding behaviors of
the organism are eliminated.  For the same organism, Bohra et
al.~\cite{bohra_identification_2018} identified a similar small set of specific
neurons whose inactivation leads to the elimination of the organism's aversive
response to bitter taste.  The term \emph{vital set} has been proposed for such
sets of neurons~\cite{Ramaswamy2019}.  A related aspect is that of identifying
subsets of neurons such that as long as neurons in such a subset remain active,
the inactivation of any other neurons outside the subset does not eliminate the
behavior.  The term \emph{degeneracy} has been proposed to describe such
phenomena (both in the setting of neuroscience, as well as in the setting of
other biological systems)~\cite{edelman_degeneracy_2001}.  In agreement with the
terminology, the term \emph{degenerate circuit} was proposed
in~\cite{Ramaswamy2019} to describe such subsets of neurons. Typically, with
respect to a given behavior, one would be interested in determining the smallest
or minimal vital or degenerate subsets of neurons, or the sizes of such sets.

Recent advances in experimental technology, especially optogenetics, have
allowed researchers to achieve precise selective activation and deactivation of
specific subsets of neurons, even those of live animals, and to record changes
in the behavior of such neurons as a result of such perturbations (see,
e.g.~\cite{zhang_closed-loop_2018,mardinly_precise_2018,forli_two-photon_2018}
for some recent advances related to these techniques).  Thus, while studying
vital or degenerate neurons for a given behavior, a researcher may be able to
collect data of the form: when a particular subset $S$ of neurons is
deactivated, is the behavior being studied still displayed?

This setting leads to an algorithmic question: given a description of the
\nckt{} and of the behavior under study (which will typically be encoded as the
eventual activation of some output neurons), determine the properties of the
vital or degenerate sets of neurons for that behavior.  A particular algorithm
for this problem may follow the above strategy of reading out the behavior of
the \nckt{} in response to selective activation and deactivation of specific
neurons.

Given the large number of perturbations possible, it is not surprising, however,
that the spectre of combinatorial explosion does cloud this strategy: we briefly
mention three works in this direction here.  Koch~\cite{koch_modular_2012}
highlighted combinatorial explosion as being a roadblock in understanding the
behavior of general biological systems with heterogeneous components.  In 2013,
Vlachos et al.~\cite{vlachos_neural_2013} proposed a ``prediction and
identification challenge'', where they invited readers to determine the
functionalities of synthetic neural circuits using a set of allowed
observations.  Kumar et al.~\cite{kumar_challenges_2013} noted the possibility
of combinatorial explosion specifically in the setting of perturbative
experiments in neuroscience by ``selective modulation'', especially pointing to
the recurrent nature of \nckts{} as a possible source of difficulty.  We remark
here that by its nature, a \nckt{} is \emph{recurrent} in the sense that the
activation state of a neuron at a given time can depend upon its own state at a
previous time, and it is also in this aspect that models of \nckts{} differ
qualitatively from the usual circuit models in computational complexity theory.

In a recent paper, Ramaswamy~\cite{Ramaswamy2019} took a step towards studying
this algorithmic difficulty in the context of computational theory.  Starting
with formal notions of \emph{vital sets} and \emph{degenerate sets} in the
context of models of \nckts{}, he formalized the following problems (here,
following Ramaswamy's notation, ``$k$-vital set'' denotes a vital set of size
exactly $k$; we defer formal definition to \cref{sec:preliminaries}).  Given as
an input a description of the \nckt{}, determine
\begin{enumerate}
\item whether there is degenerate circuit of size $k$. 
\item whether there is a minimal degenerate circuit of size $k$.
\item a minimum size degenerate circuit.
\item the set of $1$-vital sets.
\item whether there is minimal $k$-vital set.\label{prob:min-k-vital}
\item the number of minimal $k$-vital sets.\label{prob:count-min-k-vital}
\end{enumerate}
The main result of Ramaswamy's paper~\cite{Ramaswamy2019} is that each of these
problems is NP-hard.

\subsection{Our contributions}
\label{sec:our-contributions}
We show that problems 1-4 above are in fact PSPACE-hard. We also show that it is PSPACE-hard to find $(c\cdot \log n)$-approximate minimum degenerate circuit, for any constant $c$ and $n$ is size of neural circuit. In fact, the PSPACE
hardness for these problems turns out to be an immediate corollary of our main
result, which establishes the PSPACE hardness for a superficially much simpler
problem: that of \emph{simulating} a \nckt{}.

More specifically, we use a standard, synchronous, discrete-time model of
\nckts{} where each neuron is a vertex in a directed graph (which is \emph{not}
necessarily acyclic, in order to account for the recurrent nature of \nckts{}).
At each time $t$, every neuron computes a Boolean function of the activation
states at time $t-1$ of those neurons from which it receives an incoming edge
(in other words, we have a uniform \emph{conduction delay} of $1$).  We further
restrict these Boolean function to be threshold functions with coefficients and
threshold potential coming from a fixed, constant size set of small
integers. Such neurons are referred to as \emph{threshold} neurons in the
literature.  Further, the \nckt{} has a specified input neuron $I$ with no
incoming edges, and a specified output neuron $O$ with no outgoing edges.  The
\nSimulate{} problem asks: Given a \nckt{} $C$ as above as input, and given the
initial state where at time $t = 0$, neuron $I$ is \emph{stimulated} (i.e., set
to $1$), and all other neurons in $C$ are not stimulated (i.e., set to $0$), is
there a future time $t \geq 0$ when $O$ becomes stimulated (i.e., set to $1$)?
Our main result then says:
\begin{theorem}\label{thm:main:intro}
  The problem \nSimulate{} is PSPACE-complete.
\end{theorem}
As stated above, this result holds even when each neuron is restricted to be a threshold neuron with constant coefficients and threshold potential.  Further, it also holds when each neuron in the input circuit is
constrained to have at most $6$ connections.  We also emphasize that the
recurrent nature of the \nckts{} is the main feature underlying the result.

A more formal description of the result and the required technical definitions
can be found in \cref{sec:preliminaries}.  The proof of the main theorem appears
in \cref{sec:comp-hardn-simul} (as the proof of
\cref{thm:th-neural-ckt-simulation-is-PSPACE-hard}).  The PSPACE-hardness of
problems 1-4 in Ramaswamy's list, claimed above, is an easy consequence of the
PSPACE-hardness of the \nSimulate{} problem.  For completeness, we provide the
proofs at the end of \cref{sec:comp-hardn-simul}.

\subsection{Related work}
We conclude the introduction with a brief description of some related work.
Various models have been proposed to model the computational aspects of
biological neurons; we refer to the book chapter by Koch, Mo and
Softky~\cite{koch02:_handb_brain_theor_neural_networ} for a concise comparison
of various such models.  An important class among these is of models based on
thresholds, where each neuron computes either a linear or polynomial threshold
function at each time.  Low-degree polynomial threshold neurons (known as
\emph{sigma-pi} neurons in the literature) were proposed to improve upon linear
threshold neurons (also known as \emph{McCulloch–Pitts} neurons) with respect to
the modeling of ``dendritic trees'': we refer to the survey of Mel~\cite[section
4.4.4]{mel_information_1994} for further discussion of the motivation behind
different threshold models.

A particular linear threshold model that has been extensively studied is the
\emph{spiking neuron model}~\cite{maass1996lower,maass1997networks}.  This is a
continuous time model where each neurons performs a weighted integration (over
time) of ``spikes'' of activation that it observes from other neurons it is
connected to, and fires only at those times when the integral crosses a
threshold.  Maass~\cite{maass1996lower} showed that this model can be immensely
powerful: for a given $d$, there is a \emph{fixed} spiking neural network $N(d)$, such that for \emph{any} Turing machine $M$ with $d$ tapes, there is a
suitable rational assignments of weights that each neuron in $N(d)$ uses to
weigh its different neurons, such that with the corresponding weights, $N(d)$
can simulate $M$ in the sense that it can encode the output of $M$ in the
(continuous time) timing of its spike activity (we refer
to~\cite{maass1996lower} for further details).  Note, however, that this is in
sharp contrast to the model studied in our paper: here, the conduction delays
between neurons are uniform and fixed to be $1$, and further, the weights used
by the neurons are drawn from a fixed, constant size set.  The spiking neuron
model has also been the subject of some recent work dealing with questions of
asynchronous computation~\cite{hitron2020computational,hitron2019counting}.

Finally, we refer to the paper by Schmitt~\cite{schmitt_computing_1998} for
theoretical comparisons between the power of various models of neurons.

\section{Preliminaries}
\label{sec:preliminaries}
\begin{definition}[\textbf{\Nckts}]
  A \emph{\nckt} is a directed graph $G = (V, E)$, with each vertex $v \in V$
  equipped with a Boolean function
  $f_v: \inb{0, 1}^{\parent{v}} \rightarrow \inb{0,1}$. Here, $\parent{v}$
  denotes the set of those vertices $u$ in $V$ for which the directed edge
  $(u, v)$ is present in $E$.  Following standard convention we also refer
  to the vertices as \emph{neurons}. Each \nckt\ has two specially designated
  neurons: the \emph{input neuron} $I$ and the \emph{output neuron} $O$.  The
  \emph{state} of a \nckt\ at time $t$ is a specification of a bit
  $v(t) \in \inb{0,1}$ for each neuron $v$ in the circuit.  A neuron $v$ is said
  to be \emph{stimulated} at time $t$ if $v(t) = 1$, and \emph{non-stimulated}
  otherwise.
\end{definition}

We now proceed to define the dynamics of a \nckt.  At time $t = 0$, we set the
state of the circuit such that $I(0) = 1$ for the input neuron $I$ and
$v(0) = 0$ for all other neurons $v$.  For $t \geq 0$, the state of the circuit
at time $t+1$ is obtained by each vertex $v$ evaluating its function $f_v$ based
on inputs at time $t$.  Formally, given $v \in V$, suppose
$\parent{v} = \{u_1, u_2, \dots, u_k\}$.  Then we have
\begin{equation}
  v(t+1) = f_v\inp{u_1(t), u_2(t),\dots,u_k(t)}.\label{eq:1}
\end{equation}
In other words, the \nckts\ we consider have a \emph{conduction delay} of one
unit time on every edge: the stimulation state of any neuron $v$ at a given time
instant $t$ is available to all other neurons that have $v$ as a parent at time
$t + 1$.  We will therefore often write the update equation (\cref{eq:1}) in the
following abbreviated form:
\begin{equation}
  \label{eq:2}
  v \leftarrow f_v(u_1, u_2, \dots, u_k).
\end{equation}

\begin{definition}[{\textbf{\Ntv{} \nckts}}]\label{def:ntv-ckts}
  We say that a neural circuit $G=(V, E)$ is \emph{\ntv} if, starting from the
  configuration $\sigma_0$ at time $t = 0$ in which only the input node $I$ is
  stimulated (i.e., $\sigma_0(I) = 1$ and $\sigma_0(v) = 0$ for all
  $v \neq I$), there is a time $t > 0$ at which the output node is
  stimulated, i.e., $O(t) = 1$.
\end{definition}

\begin{definition}[\textbf{Threshold Function}]
\label{def:threshold-function}
 A Boolean function $f: \inb{0, 1}^k \rightarrow \inb{0, 1}$ is said to be a
\emph{threshold function} if there are (possibly negative) integers
$b, w_1, w_2, \dots, w_k$ such that
\begin{displaymath}
  f(x_1, x_2, \dots, x_k) = \I{\sum_{i=1}^kw_i x_i \geq b} \defeq
  \begin{cases}
    1 & \text{ when $\sum_{i=1}^kw_i x_i \geq b$}.\\
    0 & \text{ when $\sum_{i=1}^kw_ix_i < b$}.
  \end{cases}
\end{displaymath}
The parameter $b$ is referred to as the \emph{threshold potential} of $f$, while
the parameters $w_1, w_2, \dots, w_k$ are referred to as its \emph{weights}.  By
a slight abuse of terminology, we will use the phrase ``$f$ has weights of
absolute value at most $W$'' to signify that
$\abs{b}, \abs{w_1}, \dots, \abs{w_k} \leq W$.
\end{definition}

\paragraph{Threshold \nckts.}
\label{sec:threshold-nckts}
In any realistic model of \nckts{}, the Boolean functions that each neuron is
allowed to compute at a given time step must be suitably constrained.  In this
paper, we constrain the neurons in our \nckts{} to be compute only threshold function.  
In our computational complexity results, we will further restrict the neurons
appearing in our reduction to only draw their weights from a fixed,
constant sized set. 

A \nckt\ is said to be a \emph{threshold \nckt} if, for every neuron
$v$ in the \nckt, the corresponding function $f_v$ is a threshold
function.  We note here the well known fact that well known Boolean \textsc{And} and \textsc{OR}
functions can be represented by threshold functions
with small weights.  In particular, we have
 $x_1 \lor x_2 = [x_1 + x_2 \geq 1]$,
and $x_1 \land x_2 = [x_1 + x_2 \geq 2]$.

We briefly remark on two aspects of the \nckt{} model considered in this paper.
The first is that all conduction delays are uniformly set to $1$.  This is in
contrast, e.g., to the spiking neuron model where the conduction delays can vary
over edges.  Since our goal in this paper is to establish complexity theoretic
hardness results for algorithmic problems on \nckts, the use of a simplified
model of conduction delay only serves to strengthen our results: we show that
these algorithmic problems remain hard even for simplified model of \nckts\ that
we consider.

The second aspect is regarding comparisons with the usual Boolean circuit model
in computational complexity theory.  As discussed in the introduction, the
hardness results for \nckt{} arise as consequence of these circuits being
recurrent (the directed graph underlying a \nckt{} is not required to be
acyclic).  This is the main point of departure from the usual Boolean circuit
model, and accounts for the significantly higher computational complexity of
problems addressing \nckt{} models (e.g., the simulation problem for Boolean
circuits is trivially in $P$, in contrast to the PSPACE-hardness result for the
simulation of \nckts{} proven here).

\subsection{The problems}

Finally, we give a formal description of the algorithmic problems studied in
this paper.

\begin{problem}[\nSimulate]
  \label{ques:neural-circuit-simulation}
  \textbf{INPUT:} A threshold \nckt\ $G = (V, E)$ with input node $I$
  and output node $O$ along with the threshold functions $f_v$ at the
  neurons $v \in V$.  The weights of the threshold
  functions  are integers, and are
  drawn from a constant size set fixed in advance.

  \noindent \textbf{OUTPUT:}

  YES: if $G$ is non-trivial. That is, if, with the stimulation state in which
  only $I$ is stimulated (and none of the other nodes are stimulated), the
  neural circuit can reach a state in which the node $O$ is stimulated.

  NO: otherwise.
\end{problem}

Before describing the other problems, we recall the formalization of the notions
of \emph{vital sets} and \emph{degenerate circuits} due to
Ramaswamy~\cite{Ramaswamy2019}.  We first formalize the notion of
\emph{deactivating} or \emph{silencing} a subset of neurons in a \nckt{}.
\begin{definition}[\textbf{Silencing of a subset of neurons of a
    \nckt{}}] \label{def:silencing-neurons} Given a \nckt\ $G = (V, E)$ and a subset
  of neurons $S$, silencing of the set $S$ means that for all $t\geq 0$, and for
  all neurons $v\in S$, $\sigma_t(v)=0$, i.e., neurons in $S$ never stimulate.
\end{definition}

We now recast Ramaswamy's definitions in our terms.  Recall that a \ntv{}
\nckt{} (\Cref{def:ntv-ckts}) is one in which starting from the initial condition
in which only the input neuron $I$ is stimulated, there is a future time $t$ at
which the output neuron $O$ gets stimulated.

\begin{definition}[\textbf{Degenerate Circuit}\cite{Ramaswamy2019}]
\label{def:degenerate-circuit}
Given a \nckt\ $G = (V, E)$, a set $N \subseteq V$ of neurons is said to constitute
a degenerate circuit for $G$ if the circuit obtained by silencing the neurons in
$V \setminus N$ is \ntv{} if and only if $G$ is \ntv{}.  We assume that the
input neuron $I$ and the output neuron $O$ are always contained in any degenerate
circuit.  A \emph{minimal} degenerate circuit for $G$ is a degenerate circuit
$N$ of $G$ such that no proper subset $N'$ of $N$ forms a degenerate circuit of
$G$.  A \emph{minimum} degenerate circuit for $G$ is a degenerate circuit of $G$
of minimum size.
\end{definition}
A degenerate circuit gives us the notion of a sub-circuit of a \ntv{} \nckt{}
$C$ that is capable of showing the same behavior as $C$. Note also that given a
\nckt{} $C$ there always exists at least one degenerate circuit $C$, which is
$C$ itself.

Note also that given a \nckt{}, finding a minimum degenerate circuit is
equivalent to finding a sub-circuit of the smallest size which is capable of
showing the same behavior as the original circuit.

\begin{definition}[\textbf{Vital Set} \cite{Ramaswamy2019}]
\label{def:vital-set}
Given a \nckt{} $G = (V, E)$, a set of neurons $S\subseteq V \setminus \inb{I, O}$ is
said to be a vital set of neurons if it has a non-empty intersection with every
degenerate circuit of the neural circuit $G$.  A vital set of size $k$ is called
a $k$-vital set.  A \emph{minimal} vital set $S$ is a vital set of $G$ such that
no proper subset $S'$ of $S$ is a vital set of $G$. A \emph{minimum} vital set
is a vital set of minimum size.
\end{definition}
Note that it follows from the definition that if $S$ is a nonempty vital set of a \ntv{}
\nckt{}, then silencing the nodes in $S$ will ensure that starting from the
initial condition in which only the input node $I$ is set to $1$, the output
node $O$ will never stimulate.  We can now list the computational problems
formalized by Ramaswamy~\cite{Ramaswamy2019}, as described informally in the
introduction.

\begin{problem}[\textsc{$k$-Degenerate-Circuit}]
  \label{ques:k-size-degenerate-circuit}
  \textbf{INPUT :} A neural circuit $G = (V, E)$ with an input neuron $I\in V$,
  an output neuron $O \in V$ and a positive integer $k$ ($k\geq 2$).
  
  \textbf{OUTPUT :} A degenerate circuit $N$ of $G$ of size $k$.
\end{problem}

\begin{problem}[\textsc{Minimal-Degenerate-Circuit}]
  \label{ques:minimal-degenerate-circuit}
  \textbf{INPUT :} A neural circuit $G = (V, E)$ with an input neuron $I\in V$, an
  output neuron $O \in V$.
  
  \textbf{OUTPUT :} A minimal degenerate circuit $N$ of $G$.
\end{problem}

\begin{problem}[\textsc{Minimum-Degenerate-Circuit}]
  \label{ques:minimum-degenerate-circuit}
  \textbf{INPUT :} A neural circuit $G = (V, E)$ with an input neuron $I\in V$ and
  an output neuron $O \in V$.
  
  \textbf{OUTPUT :} A minimum degenerate circuit $N$ of $G$.
\end{problem}

\begin{problem}[\textsc{$k$-Vital-Set}]
  \label{ques:vital-set}
  \textbf{INPUT :} A neural circuit $G = (V, E)$ with an input neuron $I\in V$ and
  an output neuron $O \in V$ and a number $k$.
  
  \textbf{OUTPUT :} A $k$-vital set $N$ of $G$.
\end{problem}

\begin{problem}[\textsc{1-Vital-Sets}]
  \label{ques:1-vital-set}
  \textbf{INPUT :} A neural circuit $G = (V, E)$ with an input neuron $I\in V$ and
  an output neuron $O \in V$.
  
  \textbf{OUTPUT :} The set of neurons $N\subseteq V$, such that for every
  $v \in N$, the set $\inb{v}$ is a $1$-vital set of $G$, and such that for
  every $v \not\in N$, the set $\inb{v}$ is not a $1$-vital set of $G$.
\end{problem}

\begin{problem}[\textsc{$k$-Degenerate-Circuit-Decision}]
  \label{ques:k-size-degenerate-circuit-decision}
  \textbf{INPUT :} A neural circuit $G = (V, E)$ with an input neuron $I\in V$ and
  an output neuron $O \in V$.
  
  \noindent \textbf{OUTPUT:}

  YES: if there exists a degenerate circuit of $G$ of size $k$.

  NO: otherwise.
\end{problem}

\begin{problem}[\textsc{Minimal-Degenerate-Circuit-Decision}]
  \label{ques:minimal-degenerate-circuit-decision}
  \textbf{INPUT :} A neural circuit $G = (V, E)$ with an input neuron $I\in V$, an
  output neuron $O \in V$.
  
  \noindent \textbf{OUTPUT:}

  YES: if there exists a minimal degenerate circuit of $G$ of size at least $3$.

  NO: otherwise.
\end{problem}

\begin{problem}[\textsc{Minimum-Degenerate-Circuit-Decision}]
  \label{ques:minimum-degenerate-circuit-decision}
  \textbf{INPUT :} A neural circuit $G = (V, E)$ with an input neuron $I\in V$ and
  an output neuron $O \in V$.
  
  \noindent \textbf{OUTPUT:}

  YES: if the size of any minimum degenerate circuit of $G$ is at least $3$.

  NO: otherwise.
\end{problem}

\begin{problem}[\textsc{1-Vital-Sets-Decision}]
  \label{ques:1-vital-set-decision}
  \textbf{INPUT :} A neural circuit $G = (V, E)$ with an input neuron $I\in V$ and
  an output neuron $O \in V$.
  
  \noindent \textbf{OUTPUT:}

  YES: if the size of the set of 1-vital sets of $G$ is non-empty, i.e., if
  there is a vertex $v \in V \setminus \inb{I, O}$ such that every degenerate
  circuit of $G$ includes $v$.

  NO: otherwise.
\end{problem}

Finally, we record the following standard result, which will provide the source
problem for our PSPACE-hardness reductions.

\begin{theorem}[\textbf{True Quantified Boolean Formula (TQBF) is PSPACE-complete}]
\label{def:tqbf}
The following problem is PSPACE-complete.

\textbf{INPUT:} A fully quantified Boolean formula
\[
  \exists x_n \forall x_{n-1} \exists x_{n-2} \dots \forall x_2 \exists x_1 \phi(x_1, x_2, \dots, x_n)
\]
where $\phi(x_1, x_2, \dots, x_n)$ is a $3$-CNF formula in $n$ variables such
that each variable $x_i$ appears at most $4$ times in $\phi$.  We also assume
that $n$ is \emph{odd}, and that the quantifiers alternate strictly: the
$i$\textsuperscript{th} quantification is $\exists x_i$ if $i$ is odd, and
$\forall x_i$ if $i$ is even.  We further assume that
$\phi(x_1, x_2, \dots, x_n) = 0$ when $x_1 = x_2 = \ldots = x_n = 0$.

\textbf{OUTPUT:} YES: if the input quantified Boolean formula is true, NO
otherwise.
\end{theorem}

We briefly mention the standard methods by which the constraints assumed for the
TQBF instance can be enforced.  The constraint that each variable appears a
bounded number of times in $\phi$ can be enforced in the same manner as that for
3-SAT~\cite{tovey1984simplified}. To ensure that $\phi$ evaluates to false when
all variables are set to $0$, we replace $\phi$ by
$\phi' \defeq \phi\land (\overline{x_{n+1}} \lor \overline{x_{n+1}} \lor
\overline{x_{n+1}})$, where $x_{n+1}, x_{n+2}$ and $x_{n+3}$ are fresh variables
not appearing in $\phi$, and introduce existential quantification over these
fresh variables.  The resulting TQBF instance is true if and only if the
original instance was true (note that the position at which the existential
quantifiers on $x_{n+1}, x_{n+2}$ and $x_{n+3}$ are introduced does not matter
for this deduction).  Finally, to ensure that $n$ is odd and that the
quantifiers alternate strictly, we add extra dummy variables (which do not
appear in $\phi'$) along with the requisite quantifiers so as to enforce strict
alternation.

\section{Computational hardness of simulating neural circuits}
\label{sec:comp-hardn-simul}
In this section, we prove our main result: the PSPACE-hardness of the problem
\nSimulate{}.  We start with a description of a simple counter gadget that will
be useful in the reduction.

\subsection{The counter threshold \nckt{}}
\label{sec:counter-nckt}
In this subsection, we will give a construction of a bounded-degree threshold
counter \nckt{} which satisfies the following requirements: given a positive
integer $n$, we require a gadget with an input neuron $I$ and a set of $n$
specified neurons such that when the \nckt{} is started in the initial state
where $I = 1$ and all other neurons are set to $0$, the activation states of the
$n$ specified neurons go through, in sequence, each of the binary integers from
$0$ to $2^{n}-1$, possibly after a warm-up time known in advance.  In
\Cref{cor:th-stimulation-of-x_i_(n+1-i)}, we show that the gadget given here has
this property.

Note also that all the neurons in the constructed counter gadget compute
threshold functions, and further, the weights of these threshold functions are
integers coming from the set $\{0,1,-1,-2\}$.  Further the maximum degree of any
neuron in the construction is at most $6$ (see \cref{fig:threshold-counter}).

\begin{figure}[ht]
  \centering
  \begin{tikzpicture}
\node[latent](i){$I$};
\node[latent, right=0.9 of i](x20'){$x_{2,0}'$};
\node[latent, below=2.0 of i](x10){$x_{1,0}$};
\node[latent, below left=0.8 and 0.5 of i](x10'){$x_{1,0}'$};
\node[latent, below=2.0 of x20'](x20){$x_{2,0}$};

\node[latent, right=1.0 of x20'](y2'){$y_2'$};
\node[latent, right=1.0 of y2'](y2){$y_2$};
\node[latent, below left= 0.8 and 0.5 of y2](a2){$a_2$};
\node[latent, below right=0.8 and 0.5 of y2](b2){$b_2$};
\node[latent, below=2.0 of y2](x30){$x_{3,0}$};

\node[latent, right=1.5 of y2](y3'){$y_3'$};
\node[latent, right=1.0 of y3'](y3){$y_3$};
\node[latent, below left=0.8 and 0.5 of y3](a3){$a_3$};
\node[latent, below right=0.8 and 0.5 of y3](b3){$b_3$};
\node[latent, below=2.0 of y3](x40){$x_{4,0}$};

\node[latent, right=1.1 of y3](y4'){$y_4'$};
\node[latent, right=2.0 of y4'](yn1'){$y_{n-1}'$};
\node[latent, right=1.0 of yn1'](yn1){$y_{n-1}$};
\node[latent, below left=0.8 and 0.5 of yn1](an1){$a_{n-1}$};
\node[latent, below right=0.8 and 0.5 of yn1](bn1){$b_{n-1}$};
\node[latent, below=2.0 of yn1](xn0){$x_{n,0}$};

\edge{i}{x10};
\draw[<->] (x10) -- (x10');
\edge{i}{x20'};
\draw[<->] (x20') -- (x20);
\draw[->] (x20) to [out=0,in=270] (y2');

\edge{y2'}{y2};
\edge{y2}{a2};
\edge{y2}{b2};
\draw[<->] (a2) -- (x30);
\draw[<->] (b2) -- (x30);

\draw[->] (x30) to [out=0,in=270] (y3');
\edge{y2}{y3'};
\edge{y3'}{y3};
\edge{y3}{a3};
\edge{y3}{b3};
\draw[<->] (a3) -- (x40);
\draw[<->] (b3) -- (x40);

\draw[->] (x40) to [out=0,in=270] (y4');
\edge{y3}{y4'};
\node[right=0.3 of y4'](y4){};
\node[left=0.3 of yn1'](yn2){};
\draw[dotted,-] (y4) -- (yn2);
\node[right=0.5 of b3](a3){};
\node[left=0.5 of an1](bn2){};
\draw[dotted,-] (a3) -- (bn2);
\node[below=0.3 of yn1'](yn2a){};
\edge{y4'}{y4};
\edge{yn2}{yn1'};
\edge{yn2a}{yn1'}

\edge{yn1'}{yn1};
\edge{yn1}{an1};
\edge{yn1}{bn1};
\draw[<->] (an1) -- (xn0);
\draw[<->] (bn1) -- (xn0);

\node[right=1.0 of x40](x50){};
\node[left=1.0 of xn0](xn'){};
\draw[dotted,-] (x50) -- (xn');

\node[latent, below=0.7 of x10](x01){$x_{1,1}$};
\node[latent, below=0.7 of x01](x02){$x_{1,2}$};
\node[latent, below=0.7 of x02](x03){$x_{1,3}$};
\node[latent, below=0.7 of x03](x04){$x_{1,4}$};
\node[latent, below=0.7 of x04](x05){$x_{1,5}$};
\node[latent, below=0.7 of x05](x06){$x_{1,6}$};
\node[latent, below=3.0 of x06](x02n-1){$x_{1,2n-1}$};
\node[latent, below=1.0 of x02n-1, scale=1.4](x02n){$x_{1,2n}$};
\edge{x10}{x01};
\edge{x01}{x02};
\edge{x02}{x03};
\edge{x03}{x04};
\edge{x04}{x05};
\edge{x05}{x06};
\draw[dotted,->] (x06) -- (x02n-1);
\edge{x02n-1}{x02n};

\node[latent, right=0.9 of x01](x11){$x_{2,1}$};
\node[latent, below=0.7 of x11](x12){$x_{2,2}$};
\node[latent, below=0.7 of x12](x13){$x_{2,3}$};
\node[latent, below=0.7 of x13](x14){$x_{2,4}$};
\node[latent, below=0.7 of x14](x15){$x_{2,5}$};
\node[latent, below=0.7 of x15](x16){$x_{2,6}$};
\node[latent, below=3.0 of x16](x12n-1){$x_{2,2n-1}$};
\node[latent, below=1.0 of x12n-1,scale=1.4](x12n){$x_{2,2n}$};
\edge{x20}{x11};
\edge{x11}{x12};
\edge{x12}{x13};
\edge{x13}{x14};
\edge{x14}{x15};
\edge{x15}{x16};
\draw[dotted,->] (x16) -- (x12n-1);
\edge{x12n-1}{x12n};

\node[latent, right=2.7 of x13](x21){$x_{3,1}$};
\node[latent, below=0.7 of x21](x22){$x_{3,2}$};
\node[latent, below=0.7 of x22](x23){$x_{3,3}$};
\node[latent, below=0.7 of x23](x24){$x_{3,4}$};
\node[latent, below=3.0 of x24](x22n-3){$x_{3,2n-3}$};
\node[latent, below=1.0 of x22n-3](x22n-2){$x_{3,2n-2}$};
\edge{x30}{x21};
\edge{x21}{x22};
\edge{x22}{x23};
\edge{x23}{x24};
\draw[dotted,->] (x24) -- (x22n-3);
\edge{x22n-3}{x22n-2};

\node[latent, right=3.2 of x23](x31){$x_{4,1}$};
\node[latent, below=0.7 of x31](x32){$x_{4,2}$};
\node[latent, below=3.0 of x32](x32n-5){$x_{4,2n-5}$};
\node[latent, below=1.0 of x32n-5](x32n-4){$x_{4,2n-4}$};
\edge{x40}{x31};
\edge{x31}{x32};
\draw[dotted,->] (x32) -- (x32n-5);
\edge{x32n-5}{x32n-4};

\node[latent, right= 5.4 of x32n-5,scale=1.5](xn3){$x_{n,3}$};
\node[latent, below=1.0 of xn3,scale=1.5](xn4){$x_{n,4}$};
\node[latent, above=0.5 of xn3](xn2){$x_{n,2}$};
\node[latent, above=0.5 of xn2](xn1){$x_{n,1}$};
\edge{xn0}{xn1}
\edge{xn1}{xn2};
\edge{xn2}{xn3};
\edge{xn3}{xn4};

\node[above left= 0.5 and 0.25 of x02n](a){};
\node[below right= 0.5 and 0.25 of xn4](b){};
\draw[thick,dotted] (a) rectangle (b);
\end{tikzpicture}
  \caption{Threshold counter \nckt{}}
  \label{fig:threshold-counter}
\end{figure}

\textbf{Construction of the threshold counter \nckt{} counting 0 to $2^n-1$ (see
  \cref{fig:threshold-counter})} The counter gadget has an input neuron $I$, and
$n$ variable neurons $x_{1,0}, x_{2,0}, \ldots , x_{n,0}$. For
$2\leq i \leq n-1 $, we further introduce neurons $y_i', y_i, a_i$ and $b_i$.
Finally, we introduce two sets of auxiliary neurons: we first have the neurons
$x_{1,0}'$, $x_{2,0}'$, $x_{1,1}$, $x_{1,2}, \ldots , x_{1,2n}$, and also the
auxiliary neurons $x_{i,1}, x_{i,2},\ldots , x_{i,2n+4-2i}$, for each integer
$i$ satisfying $2\leq i \leq n$.

The connections in the circuit depends on the stimulation condition of neurons. If  stimulation of a neuron $v$ at any time $t>0$ depends on the stimulation state of neuron $u$ at time $t-1$, then there is an edge $(u,v)$.

Now, we describe the initial condition for the counter \nckt{}. Thereafter, we see the stimulation conditions of neurons, upon satisfaction of which neurons in the  counter \nckt{} stimulate at any time $t>0$.

\textbf{Initial condition:} At time $t = 0$, the neuron $I$ is stimulated (i.e.,
$I(0) = 1$) and all other neurons in the circuit are set to be non-stimulated
(i.e., their states are set to $0$).

\textbf{Stimulation conditions:} The neuron $I$ stimulates only at time $t = 0$:
we have $I(0) = 1$ and $I(t) = 0$ for all $t \geq 1$.  The neuron $x_{1,0}$
stimulates at time $t$ if either $I$ or $x_{1,0}'$ is stimulated at time $t-1$.
Formally,
\begin{equation}
  x_{1,0}\leftarrow I \lor {x_{1,0}'}.\label{eq:th-counter-1}
\end{equation}
$x_{1,0}'$ stimulates at time $t$ if  $x_{1,0}$ is stimulated at time $t-1$:
\begin{equation}
  x_{1,0}'\leftarrow x_{1,0}.\label{eq:th-counter-1'}
\end{equation}
$x_{2,0}'$ stimulates at time $t$ if either $I$ is stimulated or $x_{2,0}$ is
not stimulated at time $t-1$:
\begin{equation}
  x_{2,0}'\leftarrow I \lor \overline{x_{2,0}}.\label{eq:th-counter-2'}
\end{equation}
$x_{2,0}$ stimulates at time $t$ if  $x_{2,0}'$ is stimulated at time $t-1$:
\begin{equation}
  x_{2,0}\leftarrow x_{2,0}'.\label{eq:th-counter-2}
\end{equation}
$y_{2}'$ stimulates at time $t$ if  $x_{2,0}$ is stimulated at time $t-1$:
\begin{equation}
  y_{2}'\leftarrow x_{2,0}.\label{eq:th-counter-y2'}
\end{equation}
For $2< i < n$, $y_i'$ stimulates at time $t$ if both $x_{i,0}$ and $y_{i-1}$
are stimulated at time $t-1$:
\begin{equation}
    y_i' \leftarrow x_{i,0} \land y_{i-1}. \label{eq:th-counter-yi'}
\end{equation}
For $1< i < n$, $y_i$ stimulates at time $t$ if $y_i'$ is stimulated at time
$t-1$, $a_i$ stimulates at time $t$ if $y_i$ is stimulated and $x_{i+1,0}$ is
not stimulated at time $t-1$, and $b_i$ stimulates at time $t$ if $y_i$ is not
stimulated and $x_{i+1,0}$ is stimulated at time $t-1$:
\begin{align}
  y_i  & \leftarrow y_i',\label{eq:th-counter-yi}\\
  a_i  & \leftarrow y_i \land \overline{x_{i+1,0}}, \text{ and}\label{eq:th-counter-ai}\\
  b_i  & \leftarrow \overline{y_i} \land x_{i+1,0}.\label{eq:th-counter-bi}
\end{align}
For $2 < i \leq n$, $x_{i, 0}$ stimulates at time $t$ if at time $t-1$ either $a_{i-1}$ or $b_{i-1}$ is stimulated:
\begin{equation}
  x_{i, 0}\leftarrow a_{i-1}\lor b_{i-1}.\label{eq:th-counter-xi0}
\end{equation}
For $1 \leq j \leq 2n$, $x_{1,j}$ stimulates at time
$t$ if $x_{1,j-1}$ is stimulated at time $t-1$:
\begin{equation}
  x_{1,j} \leftarrow x_{1,j-1}.\label{eq:th-counter-x1j}
\end{equation}
Finally, for $2 \leq i \leq n$ and $1 \leq j \leq 2n+4-2i$, $x_{i,j}$ stimulates at time
$t$ if $x_{i,j-1}$ is stimulated at time $t-1$:
\begin{equation}
  x_{i,j} \leftarrow x_{i,j-1}.\label{eq:th-counter-xij}
\end{equation}

We start with the following simple observation.
\begin{observation}
  \label{lem: stimulation-of-x_1}
  The neuron $x_1$ stimulates for the first time at time 1, and thereafter, it
  changes its state at every time step. Formally, we have
  \begin{equation}
      x_{1,0}(t)=
      \begin{cases}
      0 & \text{when $t\geq 0$ is even},\\
      1 & \text{when $t\geq 1$ is odd}.
      \end{cases}
  \end{equation}

\end{observation}
\begin{proof}
  From \cref{eq:th-counter-1}, $x_{1,0}$ stimulates at time $t=1$, as the
  initial conditions give $I(0) = 1$ and $x_{1,0}'(0) = 0$.  Thereafter, we have
  $I(t) = 0$ for all $t \geq 1$, i.e., at any time $t>1$,
  $x_{1,0}(t)=x_{1,0}'(t-1)$. From \cref{eq:th-counter-1'}, at any time $t'>0$,
  $x_{1,0}'(t')=x_{1,0}(t'-1)$. These imply that at any time $t>1$,
  $x_{1,0}(t)=x_{1,0}(t-2)$. We say above that $x_{1, 0}(1) = 1$, and also know
  that $x_{1,0}(0)=0$ (due to the initial conditions). Thus, for any even time $t$,
  $x_{1,0}(t)=0$ and for any odd time $t$, $x_{1,0}(t)=1$.
\end{proof}

The following lemma builds upon the above observation to fully characterize the
stimulation profile of the key neurons in the gadget.

\begin{lemma}
  \label{lem:stimulation-of-x_i0-and-y_i}
  For $2 \leq i \leq n$, the neuron $x_{i, 0}$ is stimulated for the first time at time
  $t = 2^{i-1} + 2i-4$, and thereafter, changes its state after every $2^{i-1}$
  time-steps.  Formally, we have
  \begin{equation}
    x_{i, 0}(t) =
    \begin{cases}
      0 & 0 \leq t \leq 2^{i-1} + 2i - 5,\\
      1 & l\cdot 2^{i-1} + 2i - 4 \leq t \leq (l + 1)\cdot 2^{i-1} + 2i - 5;\; l \geq
      1\text{ odd,}\\
      0 & l\cdot 2^{i-1} + 2i - 4 \leq t \leq (l + 1)\cdot 2^{i-1} + 2i - 5;\; l \geq
      2 \text{ even.}
    \end{cases}\label{eq:th-counter-x_i0}
  \end{equation}
  Similarly, for $2 \leq i < n$, the neuron $y_i$ is stimulated for the first
  time at time $2^{i} + 2i - 4$, remains stimulated at the next time step
  i.e. at time $2^{i} + 2i - 3$ , and thereafter changes its state. It then
  repeats this behavior with a period of $2^{i}$. Formally, we have
  \begin{equation}
    y_i(t) \leftarrow
    \begin{cases}
      1 & l\cdot 2^{i} + 2i - 4 \leq t \leq l \cdot 2^{i} + 2i - 3 \text{ for some integer } l \geq 1,\\
      0 & \text{otherwise.}
    \end{cases}\label{eq:th-counter-y}
  \end{equation}
\end{lemma}
\begin{proof}
  We prove the lemma by induction on the value of $i$. We first prove the base case for the first part of lemma (\cref{eq:th-counter-x_i0}) and then for the second part of lemma (\cref{eq:th-counter-y}).

  \textbf{Base Case for \cref{eq:th-counter-x_i0} $(i=2)$:} From
  \cref{eq:th-counter-2'}, $x_{2,0}'$ stimulates for the first time at time
  $t = 1$, as the initial conditions give $I(0) = 1$ and $x_{2,0}(0) =
  0$. Thereafter, we have $I(t) = 0$ for all $t \geq 1$, i.e. at any time
  $t'\geq 2$, $x_{2,0}'(t')=1-x_{2,0}(t'-1)$. From \cref{eq:th-counter-2}, for
  $t>0$, $x_{2,0}(t)= x_{2,0}'(t-1)$. This implies that $x_{2,0}$ stimulates for
  the first time at time $2$. This proves the first item of
  \cref{eq:th-counter-x_i0} for the base case.  For $t > 2$,
  $x_{2,0}(t)=x_{2,0}'(t-1)=1-x_{2,0}(t-2)$. Since $x_{2,0}(1) = 0$ and
  $x_{2, 0}(2) = 1$, it therefore follows that $x_{2,0}(t)=1$ when
  $2k\leq t \leq 2k+1$ for odd $k\geq 1$, and $x_{2,0}(t)=0$ when
  $2k\leq t \leq 2k+1$ for even $k\geq 2$. This proves the second and third item
  of \cref{eq:th-counter-x_i0} for the base case.

  \textbf{Base Case for \cref{eq:th-counter-y} $(i=2)$:} From the initial
  conditions, we have $y_2(0) = y_2'(0) =0$. Using \cref{eq:th-counter-yi}, we
  therefore get $y_2(1)=0$. From \cref{eq:th-counter-yi,eq:th-counter-y2'}, we
  have, for $t\geq 2$, $y_2(t)=x_{2,0}(t-2)$. From the above base case analysis
  of \cref{eq:th-counter-x_i0}, $x_{2,0}(t-2)=1$ if and only if
  $2k \leq t-2 \leq 2k+1$ for odd $k\geq 1$, i.e., if
  $4\ell \leq t \leq 4\ell+1$ for some integer $\ell \geq 1$ (here, we put
  $k=2\ell-1$). This proves the base case for \cref{eq:th-counter-y}.

  We now proceed with the induction. Our induction hypothesis is that
  \cref{eq:th-counter-x_i0,eq:th-counter-y} are true for all
  $2 \leq i \leq k - 1$, for some $3 \leq k \leq n$.  We prove inductively then
  that they are true also for $i = k$. We start with \cref{eq:th-counter-x_i0}.

  From \cref{eq:th-counter-xi0} and the initial conditions, we get that
  $x_{k,0}(0) = x_{k,0}(1) = 0$. For any time $t\geq 2$, we have, from
  \cref{eq:th-counter-xi0,eq:th-counter-bi,eq:th-counter-ai}:
  \begin{align*}
    x_{k,0}(t) & = a_{k-1}(t-1) \lor b_{k-1}(t-1)\\
               &= (y_{k-1}(t-2) \land \overline{x_{k,0}}(t-2))
                 \lor (\overline{y_{k-1}}(t-2) \land x_{k,0}(t-2)) \\
               &= y_{k-1}(t-2) \oplus x_{k,0}(t-2) \tag{17} \label{eq:th-counter-xk0-t}
  \end{align*}
  To complete the proof, we first record the following consequence of the
  induction hypothesis.
  \begin{consequence}
    If $x_{k,0}(l\cdot 2^{k-1} + 2k -6)=$ $x_{k,0}(l\cdot 2^{k-1} + 2k -5)=a$
    for some integer $l\geq 1$ and $a\in \{0,1\}$ then $x_{k,0}(t)=1-a$,
    whenever $t$ satisfies
    $l\cdot 2^{k-1} + 2k -4 \leq t \leq (l+1)\cdot 2^{k-1} + 2k -5$.
    \label{cor:th-counter-xk0}
  \end{consequence}
  \begin{proof}
    From the induction hypothesis, we have
    \begin{equation}
      y_{k-1}(t)  = 1 \text{ if and only if } s\cdot 2^{k-1} + 2k -6\leq t\leq s\cdot 2^{k-1} +2k -5 \text{ for some integer $s\geq 1$ }.\label{eq:th-counter-consequence}
    \end{equation}
    If $x_{k,0}(l\cdot 2^{k-1} + 2k -6)=$ $x_{k,0}(l\cdot 2^{k-1} + 2k -5)=a$
    for some integer $l\geq 1$ then using
    \cref{eq:th-counter-xk0-t,eq:th-counter-consequence}, we get
    $x_{k,0}(l\cdot 2^{k-1} + 2k -4)=$ $x_{k,0}(l\cdot 2^{k-1} + 2k
    -3)=1-a$. Since $y_{k-1}(t)=0$ when
    $l\cdot 2^{k-1} + 2k -4 \leq t \leq (l+1)\cdot 2^{k-1} + 2k -7$ (from
    \cref{eq:th-counter-consequence}), we see from \cref{eq:th-counter-xk0-t},
    that $x_{k,0}$ remains in state $1-a$ till time
    $(l+1)\cdot 2^{k-1} + 2k -5$.
  \end{proof}
  Now, since $y_{k-1}$ stimulates for the first time at time $2^{k-1}+2k-6$, we
  see from \cref{eq:th-counter-xk0-t} that $x_{k,0}$ stimulates for the first
  time at time $2^{k-1}+2k-4$, i.e.,
  $x_{k,0}(2^{k-1}+2k-6)=x_{k,0}(2^{k-1}+2k-5)=0$. Now, iteratively using
  \cref{cor:th-counter-xk0} proves \cref{eq:th-counter-x_i0} for $i=k$.

  We now prove \cref{eq:th-counter-y} for $i=k$ for $3\leq k < n$. From the
  initial conditions and \cref{eq:th-counter-yi}, we have $y_k(0)=y_k(1)=0$. For
  $t\geq 2$, from \cref{eq:th-counter-yi,eq:th-counter-yi'}:
  \begin{equation}
    y_k(t) = y_k'(t-1) = y_{k-1}(t-2) \land x_{k,0}(t-2). \label{eq:th-y-update}
  \end{equation}
  Now, from the induction hypothesis,we have
  \begin{equation}
    y_{k-1}(t-2)
    = 1 \text{ if and only if } l\cdot 2^{k-1} + 2k -4\leq t\leq
    l\cdot 2^{k-1} +2k -3 \text{ for some integer $l\geq 1$}.\label{eq:th-y-induction}
  \end{equation}
  Further, since we have already established that the induction hypothesis
  implies \cref{eq:th-counter-x_i0} for $i = k$, we also have (assuming the
  induction hypothesis) that
  \begin{equation}
    x_{k,0}(t-2) = 1 \text{ if and only if } l'\cdot 2^{k-1} + 2k -2\leq t\leq
    (l'+1)\cdot 2^{k-1} +2k -3 \text{ for some odd $l'\geq 1$}. \label{eq:th-x-induction}
  \end{equation}
  Thus, for $t \geq 2$, $y_i(t) = 1$ if and only if $t$ satisfies the conditions
  in both \cref{eq:th-y-induction,eq:th-x-induction}.  A direct calculation
  (using $k \geq 3)$ then shows that this forces $l = l' + 1$.  We then get that
\begin{align*}
    y_k(t) & = 1 \text{ if and only if } (l'+1)\cdot 2^{k-1} + 2k -4\leq t\leq
             (l'+1)\cdot 2^{k-1} +2k -3 \text{ for some odd $l'\geq 1$}.
\end{align*}
The last condition can be written as
``$\alpha \cdot 2^{k} + 2k -4\leq t\leq \alpha \cdot 2^{k} +2k -3$, where
$\alpha=(l'+1)/2 \geq 1$ is an integer.''  This completes the induction.
\end{proof}

We now record three corollaries of the above lemma.
\begin{corollary}
  The neuron $x_{1,2n}$ stimulates for the first time at time $2n+1$, and thereafter changes its state at every time step. Formally,
  \begin{displaymath}
  x_{1,2n}(t)=
  \begin{cases}
      0 & 0 \leq t \leq 2n ,\\
      1 & t = k + 2n;\; k \geq
      1\text{ odd,}\\
      0 & t = k + 2n;\; k \geq
      2 \text{ even.}
    \end{cases}
  \end{displaymath}
  \label{cor:th-counter-x1-2n}
\end{corollary}
\begin{proof}
  From \cref{eq:th-counter-x1j}, we have $x_{1,2n}(t) = x_{1,0}(t-2n)$, for all $t\geq 2n$, and from the same equation and initial conditions, we have $x_{1,2n}(t)=0$ for $0\leq t \leq 2n$. Together with \cref{lem: stimulation-of-x_1}, this implies the claim of the corollary.
\end{proof}

\begin{corollary}
  \label{cor:counter-bit-values}
  For each $2 \leq i \leq n$, the neuron $x_{i,2n+4-2i}$ stimulates for the first
  time at time $2n+2^{i-1}$, and thereafter changes its state periodically with
  time period $2^{i-1}$.  Formally,
  \begin{displaymath}
    x_{i, 2n + 4 -2i}(t) =
    \begin{cases}
      0 & 0 \leq t \leq 2^{i-1} + 2n - 1,\\
      1 & l\cdot 2^{i-1} + 2n \leq t \leq (l + 1)\cdot 2^{i-1} + 2n - 1;\; l \geq
      1\text{ odd,}\\
      0 & l\cdot 2^{i-1} + 2n \leq t \leq (l + 1)\cdot 2^{i-1} + 2n - 1;\; l \geq
      2 \text{ even.}
    \end{cases}
  \end{displaymath}
\end{corollary}
\begin{proof}
  From \cref{eq:th-counter-xij}, for $2\leq i\leq n$ we have
  $x_{i, 2n + 4 - 2i}(t) = x_{i, 0}(t - 2n -4+2i)$, for all $t \geq 2n + 4 - 2i $, and
  from the same equations and the initial conditions, we have that
  $x_{i, 2n + 4 - 2i}(t) = 0$ for $0 \leq t \leq 2n + 3 - 2i$.  Together with
  \cref{eq:th-counter-x_i0} of \cref{lem:stimulation-of-x_i0-and-y_i}, this implies the
  claim of the corollary.
\end{proof}

\begin{corollary}
  \label{cor:th-stimulation-of-x_i_(n+1-i)}
  At any time $t \geq 2n$, the string $x_{n,4}x_{n-1,6}\dots x_{2,2n}x_{1,2n}$,
  when interpreted as a binary integer, is equal to the remainder obtained when
  dividing $t-2n$ by $2^n$ (i.e., $x_{1,2n}$ is the least significant bit and
  $x_{i,2n+4-2i}$ is the $i$th least significant bit in the binary
  representation of $t-2n$, for all $2 \leq i \leq n$). For any time $t < 2n$,
  the same string is a string of zeros.  In particular, for any string
  $\sigma \in \inb{0,1}^n$, there is a unique time $t$,
  $2n \leq t \leq 2n + 2^{n}-1$, such that at time $t$,
  $\sigma_1 = x_{1, 2n}(t)$ and $\sigma_i = x_{i,2n+4-2i}(t)$ for all
  $2 \leq i \leq n$.
\end{corollary}
\begin{proof}
  From \cref{cor:th-counter-x1-2n,cor:counter-bit-values},
  $x_{1, 2n}(t) = x_{i, 2n + 4 -2i}(t) = 0$ for all $t < 2n$ and
  $2 \leq i \leq n$.  From \cref{cor:counter-bit-values}, we also see that for
  $t' \geq 0$, $x_{i,2n + 4 - 2i}(t' + 2n) = 1$ (for $2\leq i\leq n)$ if and
  only if for some \emph{odd} positive integer $l$,
  $l\cdot 2^{i-1} \leq t' \leq (l+1)\cdot2^{i-1} - 1$. Also, from
  \cref{cor:th-counter-x1-2n}, for $t' \geq 0$, $x_{1,2n}(t' + 2n) = 1$ if and
  only if $t'$ is an \emph{odd} positive integer. Thus, for $t' \geq 0$,
  $x_{1,2n}(t'+2n)=1$ if and only if the least significant bit in the binary
  representation of $t'$ is 1, and for all $2 \leq i \leq n$,
  $x_{i, 2n + 4 - 2i}(t' + 2n) = 1$ if and only if the $i$th least significant
  bit in the binary representation of $t'$ is $1$.  This establishes the claim.
\end{proof}

\subsection{The reduction}
We now prove our main theorem (\Cref{thm:main:intro} in the introduction). We
begin with a more formal restatement.

\begin{theorem}
\label{thm:th-neural-ckt-simulation-is-PSPACE-hard} 
The \nameref{ques:neural-circuit-simulation} problem is PSPACE-hard, even when
restricted to \nckts{} of maximum degree at most $6$, and such that every neuron
computes a threshold function with weights of absolute value at most $2$.
\end{theorem}
\begin{proof} We reduce the TQBF problem to the bounded degree
  \nameref{ques:neural-circuit-simulation} problem. In particular,
  corresponding to any given TQBF instance, we construct a threshold \nckt{}
  (having maximum degree 6) with an input and output neuron such that starting
  with the initial condition in which only the input neuron is stimulated at
  time $t=0$, the output neuron stimulates at some time $t > 0$ if and only if
  the quantified Boolean formula is true.

  Let the given QBF formula be $\exists x_n$ $ \forall x_{n-1}$ $ \ldots$
  $ \exists x_3$ $ \forall x_2$ $ \exists x_1$ $\phi(x_1,x_2,\ldots,x_n)$, where
  $\phi(x_1,x_2,\ldots,x_n)$ is a 3-CNF formula with $m$ clauses. As stated in
  \cref{def:tqbf}, we can assume that each variable $x_i$ occurs at most $4$
  times in $\phi$, and further that the number $n$ of quantifications is odd,
  with all the odd-indexed variables quantified with a $\exists$ quantifier and
  all the even-indexed variables quantified with a $\forall$ quantifier. From
  \Cref{def:tqbf}, we also assume that $\phi$ is false for
  $x_1=x_2=\ldots=x_n=0$. Our goal now is to create (in time polynomial in the
  size of $\phi$) a threshold \nckt{} $C_\phi$ with an input neuron $I$ and an
  output neuron $O$ such that when $C_\phi$ is started in the initial state in
  which only $I$ is stimulated at time $t = 0$, the output neuron $O$ stimulates
  at some future time $t>0$ if and only if $\exists x_n$ $ \forall x_{n-1}$
  $ \ldots$ $ \exists x_3$ $ \forall x_2$ $ \exists x_1$
  $\phi(x_1,x_2,\ldots,x_n)$ is true.

  \textbf{Construction of the threshold \nckt{} corresponding to a QBF (see
    \cref{fig:th-ckt-construction}) :} We first describe the nodes in $C_\phi$.
  Of course, $C_\phi$ contains an input neuron $I$ and an output neuron $O$.
  The first important component of our construction is an $n$-variable counter
  \nckt{} (as discussed in \cref{sec:counter-nckt}).  The input neuron $I$ of
  $C_\phi$ is identified with the input neuron of the counter \nckt{}.  Further,
  the neuron $x_{1,2n}$ in the counter \nckt{} is relabelled as $x_1$ and for
  all $2 \leq i \leq n$, the node $x_{i, 2n +4 - 2i}$ in the counter \nckt{} is
  relabelled as $x_i$.  As we will see later, these relabelled neurons $x_i$,
  $1\leq i \leq n$, will correspond to the variables $x_i$ of the same name
  appearing in $\phi$.

In addition to the nodes in the counter \nckt{}, we have several auxiliary
nodes.  First, for each $1 < i \leq n$, we introduce $m + i - 1$ neurons
$z_{i,1}, \ldots , z_{i,m+i-1}$.  For odd $1\leq i \leq n$, we introduce a
neuron $s_{i,0}$ while for even $1< i < n$, we introduce three neurons $s_{i,0}', s_{i,0}$
and $s_{i,1}$.  For odd $1< i \leq n$, we introduce two neuron $p_i$ and $q_i$.  Further, for each clause $c_i$, where $1 \leq i \leq m$, in
$\phi$, we introduce a clause neuron $c_i$, of the same name.  We also
introduce, for each $1 < i \leq m$, $i-1$ auxiliary clause neurons
$d_{i,1}, \ldots , d_{i,i-1}$.

The connections in the circuit depends on the stimulation condition of
neurons. If stimulation of a neuron $v$ at any time $t>0$ depends on the
stimulation state of neuron $u$ at time $t-1$, then there is an edge $(u,v)$.
Now, we describe the initial condition for the \nckt{} $C_\phi$. Thereafter, we
see the stimulation conditions of neurons, upon satisfaction of which neurons in
the \nckt{} stimulate at any time $t>0$.

\textbf{Initial Condition:} At time $t = 0$, the neuron $I$ is stimulated and
all other neurons are set to be non-stimulated.

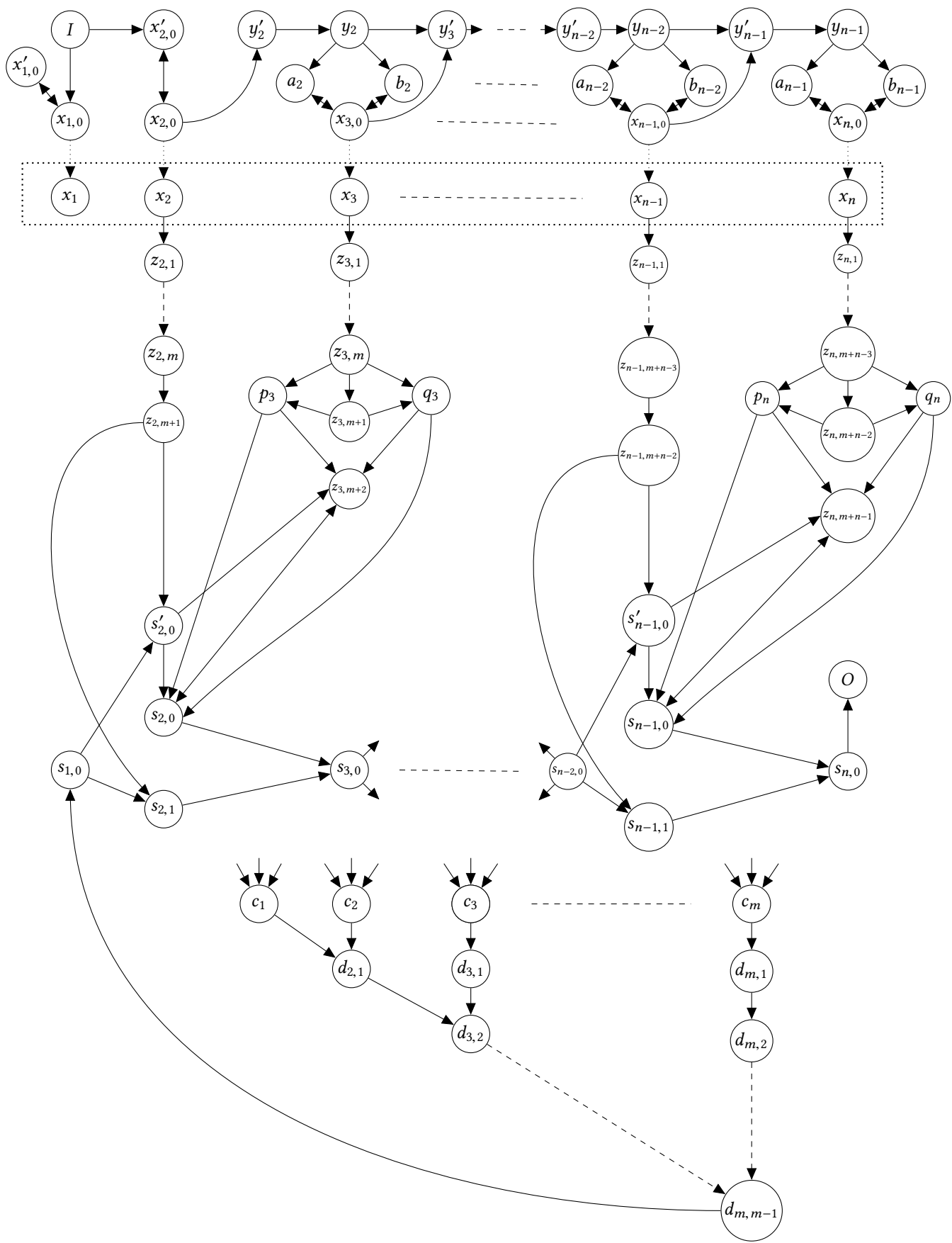
\begin{figure}[ht]
  \centering
  \begin{tikzpicture}
\node[latent](i){$I$};
\node[latent, right=1.0 of i](x20'){$x_{2,0}'$};
\node[latent, below=1.0 of i](x10){$x_{1,0}$};
\node[latent, above left=0.5 and 0.3 of x10](x10'){$x_{1,0}'$};
\node[latent, below=1.0 of x20'](x20){$x_{2,0}$};

\node[latent, right=1.0 of x20'](y2'){$y_2'$};
\node[latent, right=1.0 of y2'](y2){$y_2$};
\node[latent, below left=0.5 and 0.5 of y2](a2){$a_2$};
\node[latent, below right=0.5 and 0.5 of y2](b2){$b_2$};
\node[latent, below=1.0 of y2](x30){$x_{3,0}$};

\node[latent, right=1.1 of y2](y3'){$y_3'$};
\node[latent, right=3.0 of y3'](y3){$y_{n-2}$};
\node[latent, left=0.5 of y3](yn-2'){$y_{n-2}'$};
\node[latent, below left=0.5 and 0.5 of y3](a3){$a_{n-2}$};
\node[latent, below right=0.5 and 0.5 of y3](b3){$b_{n-2}$};
\node[latent, below=1.0 of y3,scale=0.8](x40){$x_{n-1,0}$};

\node[latent, right=1.1 of y3](yn1'){$y_{n-1}'$};
\node[latent, right=1.0 of yn1'](yn1){$y_{n-1}$};
\node[latent, below left=0.5 and 0.5 of yn1](an1){$a_{n-1}$};
\node[latent, below right=0.5 and 0.5 of yn1](bn1){$b_{n-1}$};
\node[latent, below=1.0 of yn1](xn0){$x_{n,0}$};

\edge{i}{x10};
\draw[<->] (x10) -- (x10');
\edge{i}{x20'};
\draw[<->] (x20') -- (x20);
\draw[->] (x20) to [out=0,in=270] (y2');

\edge{y2'}{y2};
\edge{y2}{a2};
\edge{y2}{b2};
\draw[<->] (a2) -- (x30);
\draw[<->] (b2) -- (x30);
\edge{y2}{y3'};

\draw[->] (x30) to [out=0,in=270] (y3');
\edge{yn-2'}{y3};
\edge{y3}{a3};
\edge{y3}{b3};
\draw[<->] (a3) -- (x40);
\draw[<->] (b3) -- (x40);
\edge{y3}{yn1'};

\draw[->] (x40) to [out=0,in=270] (yn1');
\node[right=0.3 of y3'](y4){};
\node[left=0.3 of yn-2'](yn2){};
\draw[dashed,-] (y4) -- (yn2);
\node[right=0.7 of b2](a31){};
\node[left=0.3 of a3](b31){};
\draw[dashed,-] (a31) -- (b31);
\edge{y3'}{y4};
\edge{yn2}{yn-2'};

\edge{yn1'}{yn1};
\edge{yn1}{an1};
\edge{yn1}{bn1};
\draw[<->] (an1) -- (xn0);
\draw[<->] (bn1) -- (xn0);

\node[right=1.0 of x30](x50){};
\node[left=1.0 of x40](xn'){};
\draw[dashed,-] (x50) -- (xn');

\node[latent, below=0.7 of x10] (x1) {$x_1$};
\node[latent, below=0.7 of x20] (x2) {$x_2$};
\node[latent, below=0.7 of x30] (x3) {$x_3$};
\node[latent, below=0.7 of x40, scale=0.9] (x4) {$x_{n-1}$};
\node[latent, below=0.7 of xn0] (xn) {$x_n$};

\node[right=5.5 of x1](x31){};
\node[left=4.3 of xn] (x41){};
\draw[dashed,-] (x31) -- (x41);

\draw[dotted,->] (x10) -- (x1);
\draw[dotted,->] (x20) -- (x2);
\draw[dotted,->] (x30) -- (x3);
\draw[dotted,->] (x40) -- (x4);
\draw[dotted,->] (xn0) -- (xn);

\node[above left= 0.25 and 0.5 of x1](a){};
\node[below right= 0.1 and 0.25 of xn](b){};
\draw[thick,dotted] (a) rectangle (b);

\node[latent, below=0.5 of x2](z21){$z_{2,1}$};
\node[latent, below=1.0 of z21](z2m){$z_{2,m}$};
\node[latent, below=0.5 of z2m, scale=0.75](z2m+1){$z_{2,m+1}$};

\node[latent, below=0.5 of x3](z31){$z_{3,1}$};
\node[latent, below=1.0 of z31](z3m){$z_{3,m}$};
\node[latent, below=0.5 of z3m, scale=0.75](z3m+1){$z_{3,m+1}$};
\node[latent, below=0.5 of z3m+1, scale=0.75](z3m+2){$z_{3,m+2}$};
\node[latent, below left=0.25 and 1.0 of z3m](p3){$p_3$};
\node[latent, below right=0.25 and 1.0 of z3m](q3){$q_3$};

\node[latent, below=0.5 of x4, scale=0.75](zn-11){$z_{n-1,1}$};
\node[latent, below=1.0 of zn-11, scale=0.7](zn-1m+n-3){$z_{n-1,m+n-3}$};
\node[latent, below=0.5 of zn-1m+n-3, scale=0.7](zn-1m+n-2){$z_{n-1,m+n-2}$};

\node[latent, below=0.5 of xn, scale=0.75](zn1){$z_{n,1}$};
\node[latent, below=1.0 of zn1, scale=0.75](znm+n-3){$z_{n,m+n-3}$};
\node[latent, below=0.5 of znm+n-3, scale=0.75](znm+n-2){$z_{n,m+n-2}$};
\node[latent, below=0.5 of znm+n-2, scale=0.75](znm+n-1){$z_{n,m+n-1}$};
\node[latent, below left=0.25 and 1.0 of znm+n-3, scale=0.9](pn){$p_{n}$};
\node[latent, below right=0.25 and 1.0 of znm+n-3, scale=0.9](qn){$q_{n}$};

\edge{x2}{z21};
\draw[dashed,->](z21) -- (z2m);
\edge{z2m}{z2m+1};

\edge{x3}{z31};
\draw[dashed,->](z31)--(z3m);
\edge{z3m}{z3m+1};
\edge{z3m}{p3};
\edge{z3m}{q3};
\edge{z3m+1}{p3};
\edge{z3m+1}{q3};
\edge{p3}{z3m+2};
\edge{q3}{z3m+2};

\edge{x4}{zn-11};
\draw[dashed,->](zn-11) -- (zn-1m+n-3);
\edge{zn-1m+n-3}{zn-1m+n-2};

\edge{xn}{zn1};
\draw[dashed,->](zn1)--(znm+n-3);
\edge{znm+n-3}{znm+n-2};
\edge{znm+n-3}{pn};
\edge{znm+n-3}{qn};
\edge{znm+n-2}{pn};
\edge{znm+n-2}{qn};
\edge{pn}{znm+n-1};
\edge{qn}{znm+n-1};

\node[latent, below=10.0 of x1] (s10) {$s_{1,0}$};
\node[latent, below=9.0 of x2] (s20) {$s_{2,0}$};
\node[latent, above= 1.0 of s20] (s20'){$s_{2,0}'$};
\node[latent, below =1.0 of s20] (s21){$s_{2,1}$};
\node[latent, below =10 of x3] (s30){$s_{3,0}$};

\node[latent, below=9.0 of x4] (sn-10) {$s_{n-1,0}$};
\node[latent, above= 1.0 of sn-10] (sn-10'){$s_{n-1,0}'$};
\node[latent, below =1.0 of sn-10] (sn-11){$s_{n-1,1}$};
\node[latent, below =10 of xn] (sn0){$s_{n,0}$};
\node[latent, left=4.5 of sn0, scale=0.75] (sn-20) {$s_{n-2,0}$};
\node[latent, above =1 of sn0] (o){$O$};

\edge{sn0}{o};
\edge{s10}{s20'};
\edge{s10}{s21};
\edge{z2m+1}{s20'};
\edge{s20'}{z3m+2.west};
\draw[<->] (s20) -- (z3m+2);
\edge{s20}{s30};
\edge{s21}{s30};
\edge{s20'}{s20};

\edge{p3}{s20};
\draw[->] (q3) to [out=270,in=45] (s20.east);
\draw[->] (z2m+1) to [out=180,in=135] (s21);

\edge{sn-20}{sn-10'};
\edge{sn-20}{sn-11};
\edge{zn-1m+n-2}{sn-10'};
\edge{sn-10'}{znm+n-1.west};
\draw[<->] (sn-10) -- (znm+n-1);
\edge{sn-10}{sn0};
\edge{sn-11}{sn0};
\edge{sn-10'}{sn-10};
\edge{pn}{sn-10};
\draw[->] (qn) to [out=270,in=45] (sn-10.east);
\draw[->] (zn-1m+n-2) to [out=180,in=135] (sn-11);

\node[above right=0.3 and 0.3 of s30](s30a){};
\node[below right=0.3 and 0.3 of s30](s30b){};
\edge{s30}{s30a};
\edge{s30}{s30b};

\node[above left=0.3 and 0.3 of sn-20](sn-20a){};
\node[below left=0.3 and 0.3 of sn-20](sn-20b){};
\edge{sn-20}{sn-20a};
\edge{sn-20}{sn-20b};

\node[right=0.3 of s30](s30c){};
\node[left=0.3 of sn-20](sn-20c){};
\draw[dashed,-] (s30c) -- (sn-20c);

    \node[latent, below right  =2.0 and 3  of s10] (c1) {$c_1$};
    \node[latent, right= 1.0 of c1] (c2) {$c_2$};
    \node[latent, right= 1.5 of c2] (c3) {$c_3$};
    \node[latent, right= 1.5 of c2] (c3) {$c_3$};
    \node[ right= 0.5 of c3] (c4) {};
    \node[latent, right= 4.5 of c3] (cm) {$c_m$};
    \node[ left= 0.5 of cm] (cm') {};
    \node[above right=0.5 and 0.1 of c1] (c1x){};
    \node[above left=0.5 and 0.1 of c1] (c1y){};
    \node[above= 0.5 of c1] (c1z){};
    \edge{c1x}{c1};
    \edge{c1y}{c1};
    \edge{c1z}{c1};
    \node[above right=0.5 and 0.2 of c2] (c2x){};
    \node[above left=0.5 and 0.2 of c2] (c2y){};
    \node[above= 0.5 of c2] (c2z){};
    \edge{c2x}{c2};
    \edge{c2y}{c2};
    \edge{c2z}{c2};
    \node[above right=0.5 and 0.2 of c3] (c3x){};
    \node[above left=0.5 and 0.2 of c3] (c3y){};
    \node[above= 0.5 of c3] (c3z){};
    \edge{c3x}{c3};
    \edge{c3y}{c3};
    \edge{c3z}{c3};
    \node[above right=0.5 and 0.2 of cm] (cmx){};
    \node[above left=0.5 and 0.2 of cm] (cmy){};
    \node[above= 0.5 of cm] (cmz){};
    \edge{cmx}{cm};
    \edge{cmy}{cm};
    \edge{cmz}{cm};
    \draw[dashed, -] (c4) to (cm');
    \node[latent,below =0.5 of c2] (d21){$d_{2,1}$};
    \node[latent,below =0.5 of c3] (d31){$d_{3,1}$};
    \node[latent,below =0.5 of d31] (d32){$d_{3,2}$};
    \node[latent,below =0.5 of cm] (dm1){$d_{m,1}$};
    \node[latent,below =0.5 of dm1] (dm2){$d_{m,2}$};
    \node[latent,below =2.2 of dm2] (dmm'){$d_{m,m-1}$};
    \node[right =0.5 of d32] (d42){};
    \node[below =0.5 of d42] (d43){};
    \node[below right =0.5 and 0.5 of d43] (d54){};
    \node[above left =0.5 and 0.5 of dmm'] (dmm''){};

    \edge{c1}{d21};
    \edge{c2}{d21};
    \edge{c3}{d31};
    \edge{d21}{d32};
    \edge{d31}{d32};
    \edge{cm}{dm1};
    \edge{dm1}{dm2};
    \draw[dashed,->](dm2) to (dmm');
    \draw[dashed,->](d32) to (dmm');
    \draw[->] (dmm') to [out=180,in=270] (s10);

    \end{tikzpicture}
    \caption{Threshold neural circuit for n-variable QBF}
    \label{fig:th-ckt-construction}
    \end{figure}

\textbf{Stimulation Condition:} The stimulation conditions of the neurons
appearing as part of the counter \nckt{} remain the same as those discussed
earlier.  We now describe the stimulation conditions of the other nodes.

For $1<i\leq n $, $z_{i,1}$ stimulates at time $t$ if $x_{i}$ is
stimulated at time $t-1$. Formally,
\begin{equation}
  z_{i,1}\leftarrow x_{i}.\label{eq:th-ckt-zi1}
\end{equation}

For even $1<i<n$ and $1<j\leq m+i-1$, $z_{i,j}$ stimulates at time $t$ if
$z_{i,{j-1}}$ stimulates at time $t-1$.  Formally,
\begin{equation}
  z_{i,j}\leftarrow z_{i,j-1}.\label{eq:th-ckt-even-zij}
\end{equation}

For odd $1<i\leq n$ and $1<j\leq m+i-2$, $z_{i,j}$ stimulates at time $t$ if
$z_{i,{j-1}}$ stimulates at time $t-1$. Formally,
\begin{equation}
  z_{i,j}\leftarrow z_{i,j-1}.\label{eq:th-ckt-odd-zij}
\end{equation}

For odd $1<i\leq n$,  $p_i$ stimulates at time $t$ if at time $t-1$, $z_{i,m+i-3}$ is stimulated and $z_{i,m+i-2}$ is not stimulated, and $q_i$ stimulates at time $t$ if at time $t-1$, $z_{i,m+i-3}$ is not stimulated and $z_{i,m+i-2}$ is stimulated. Formally,
\begin{align}
  p_i & \leftarrow z_{i,m+i-3} \land \overline{z_{i,m+i-2}}.\label{eq:th-ckt-pi}\\
  q_i & \leftarrow \overline{z_{i,m+i-3}} \land z_{i,m+i-2}.\label{eq:th-ckt-qi}
\end{align}

For odd $1<i\leq n$, $z_{i,m+i-1}$ stimulates at time $t$ if, at time $t-1$, at
least one of $s_{i-1,0}'$ or $s_{i-1,0}$ is stimulated and neither $p_i$ nor
$q_i$ is stimulated. Formally,
\begin{equation}
  z_{i,m+i-1}\leftarrow (s_{i-1,0}' \lor s_{i-1,0})\land \overline{p_i} \land \overline{q_i}.\label{eq:th-ckt-odd-zim+i-1}
\end{equation}
We note here a couple of simple consequences of
\cref{eq:th-ckt-zi1,eq:th-ckt-even-zij,eq:th-ckt-odd-zij,eq:th-ckt-pi,eq:th-ckt-qi},
and the initial conditions.
\begin{observation}\label{obv:th-z-even-update}
  For any $t \geq 0$, for even $1 < i < n$ and $1 \leq j \leq m + i - 1$,
  $z_{i,j}(t + j) = x_i(t)$, and $z_{i,j}(t) = 0$ for $t \leq j$.
\end{observation}

\begin{observation}\label{obv:th-z-odd-update}
  For any $t \geq 0$, for odd $1 < i \leq n$ and $1 \leq j \leq m + i - 2$,
  $z_{i,j}(t + j) = x_i(t)$, and $z_{i,j}(t) = 0$ for $t \leq j$.
\end{observation}

Using the above observations, we get the following claim, which will help us in
analyzing the update in \cref{eq:th-ckt-odd-zim+i-1}.

\begin{claim}\label{obv:th-pi-qi-update}
  For any $t \geq 0$ and for odd $1<i\leq n$,
 \begin{equation}
    \overline{p_i}(t) \land \overline{q_i}(t)
    =
    \begin{cases}
      0 & \text{if $t=l\cdot 2^{i-1} +2n +m+i-2$ for some integer $l \geq 1$}\\
      1 & \text{otherwise.}
    \end{cases}\label{eq:z-induction}
  \end{equation}
\end{claim}
\begin{proof}
  The initial condition implies that $p_i(0) = q_i(0) = 0$, so we concentrate on
  the case $t \geq 1$.  Using
  \cref{eq:th-ckt-pi,eq:th-ckt-qi,obv:th-z-odd-update}, we then have
  \begin{align}
     p_i(t)\lor q_i(t) & = (z_{i,m+i-3}(t-1) \land \overline{z_{i,m+i-2}}(t-1)) \lor (\overline{z_{i,m+i-3}(t-1)} \land z_{i,m+i-2}(t-1))\nonumber\\
     & = z_{i,m+i-3}(t-1) \oplus z_{i,m+i-2}(t-1).\label{eq:p-q-equation}
  \end{align}
  When $t \leq m + i - 2$, \Cref{obv:th-z-odd-update} implies that the right
  hand side of \cref{eq:p-q-equation} is $0$.  On the other hand, when
  $t = t' + m + i - 1$ for some $t' \geq 0$, we use \Cref{obv:th-z-odd-update}
  along with \cref{eq:p-q-equation} to get
  \begin{displaymath}
    p_i(t)\lor q_i(t) =  x_i(t' + 1) \oplus x_i(t').
  \end{displaymath}
  Now, recalling that the neuron $x_i$ is just a relabelled version of the
  neuron $x_{i, 2n + 4 - 2i}$ in the counter circuit (since $i > 1$), we see
  from \Cref{cor:counter-bit-values} that $x_i(t' + 1) \oplus x_i(t') = 1$ if
  and only if $t' = l\cdot 2^{i-1} + 2n - 1$, for some integer $l \geq 1$.  We
  thus get that $p_i(t) \lor q_i(t) = 1$ if and only if
  $t = l\cdot 2^{i-1} + 2n + m +i - 2$ for some integer $l \geq 1$.  This proves
  the claim, since $\overline{p_i(t)} \land \overline{q_i(t)}$ =
  $\overline{ p_i(t)\lor q_i(t)}$.
\end{proof}

We now continue with our description of the stimulation conditions. The clause
neuron $c_i$ stimulates at time $t$ if the stimulation states of the variable
neurons $x_j$ corresponding to the variables appearing in $c_i$ give a
satisfying assignment for the clause $c_i$ at time $t-1$.  For $2<i\leq m$,
$d_{i,1}$ stimulates at time $t$ if $c_{i}$ is stimulated at time $t-1$, while
$d_{2,1}$ stimulates at time $t$ if both $c_{1}$ and $c_2$ are stimulated at
time $t-1$.
\begin{align}
  d_{2,1} &\leftarrow c_{1}\land c_{2}\text{, and}\label{eq:th-ckt-d2}\\
  d_{i,1} &\leftarrow c_{i}\text{, when $i > 2$}.\label{eq:th-ckt-di1}
\end{align}
For $3<i\leq m$ and $1<j<i-1$, $d_{ij}$ stimulates at time $t$ if $d_{i,j-1}$ is
stimulated at time $t-1$.
\begin{equation}
  d_{i,j}\leftarrow d_{i,j-1}.\label{eq:th-ckt-dij}
\end{equation}
For $2<i\leq m$, $d_{i,i-1}$ stimulates at time $t$ if both $d_{i-1,i-2}$ and
$d_{i,i-2}$ are stimulated at time $t-1$.
\begin{equation}
  d_{i,i-1}\leftarrow d_{i-1,i-2}\land d_{i,i-2}.\label{eq:th-ckt-dilast}
\end{equation}
The neuron
$s_{1,0}$ stimulates at time $t$ if $d_{m,m-1}$ stimulates at time $t-1$.
\begin{equation}
  s_{1,0}\leftarrow d_{m,m-1}.\label{eq:th-ckt-s10}
\end{equation}
Before proceeding, we record here a simple consequence of
\cref{eq:th-ckt-d2,eq:th-ckt-dilast,eq:th-ckt-dilast,eq:th-ckt-di1,eq:th-ckt-dij,eq:th-ckt-s10}
and the initial conditions.
\begin{observation}
    \label{obv:th-d-s-update}
    For any $t \geq 0$, we have $s_{1,0}(t + m) = 1$ if and only if $c_i(t) = 1$
    for all $1 \leq i \leq m$.  Also, $s_{1,0}(t) = 0$ for all $t \leq m$.
\end{observation}

For even $1<i<n$, $s_{i,0}'$ stimulates at time $t$ if  $s_{i-1,0}$ is stimulated and $z_{i,m+i-1}$ is not stimulated at time $t-1$. Formally,
\begin{equation}
  s_{i,0}'\leftarrow s_{i-1,0} \land \overline{z_{i,m+i-1}}.\label{eq:th-ckt-si0'}
\end{equation}
For even $1<i<n$, $s_{i,0}$ stimulates at time $t$ if at least one of $s_{i,0}'$
and $z_{i+1,m+i}$ is stimulated and neither of $p_{i+1}$ and $q_{i+1}$ is
stimulated at time $t-1$. Formally,
\begin{equation}
  s_{i,0}\leftarrow (s_{i,0}' \lor z_{i+1,m+i}) \land \overline{p_{i+1}} \land \overline{q_{i+1}}.\label{eq:th-ckt-si0-even}
\end{equation}
For even $1<i<n$, $s_{i,1}$ stimulates at time $t$ if both $s_{i-1,0}$ and
$z_{i,m+i-1}$ are stimulated at time $t-1$. Formally,
\begin{equation}
  s_{i,1}\leftarrow s_{i-1,0}\land {z_{i,m+i-1}}.\label{eq:ckt-si1-even}
\end{equation}
For odd $1<i \leq n$, $s_{i,0}$ stimulates at time $t$ if both $s_{i-1,0}$ and
$s_{i-1,1}$ are stimulated at time $t-1$. Formally,
\begin{equation}
  s_{i,0}\leftarrow s_{i-1,0}\land {s_{i-1,1}}.\label{eq:ckt-si0-odd}
\end{equation}
Finally, $O$ stimulates at time $t$ if $s_{n,0}$ stimulates at time $t-1$.
\begin{equation}
  O \leftarrow s_{n,0}.\label{eq:ckt-O}
\end{equation}
Note that the size of $C_\phi$ is $\poly{n}$, and the above description of $C_\phi$ can be constructed in time $\poly{n}$ given $\phi$ as input.

Observe that all neurons in the above circuits have small degree and have only
threshold gates.  We record this formally in the following observation.

\begin{observation}
\label{cor:th-our-ckt-is-threshold-ckt}
The update function of all neurons in $C_\phi$ as constructed above are
threshold update functions with weights of absolute value at most $2$.  Further
each neuron has degree at most $6$.
\end{observation}
\begin{proof}
  Except for the neurons $z_{i,m+i-1}$ for odd $1<i\leq n$, $s_{i,0}$ for even
  $1<i<n$, and the clause neurons $c_i$ for $1\leq i \leq m$, all other neurons
  in \cref{fig:th-ckt-construction} either copy the stimulation state of some
  neuron (e.g., in \cref{eq:th-ckt-zi1}, $z_{2,1}$ copies the stimulation state
  of $x_2$ ) or compute conjunctions or disjunctions of stimulation states of
  two neurons.  All of these are easily seen to be threshold functions with
  weights of absolute values at most $2$. The clause neurons $c_i$ compute a
  conjunction of three other neurons (possibly negated), and this also is a
  threshold function with weights coming from the same set.  We now consider the
  remaining neurons ($z_{i,m+i-1}$ for odd $1<i\leq n$ and $s_{i,0}$ for even
  $1<i<n$).

  We begin by noting that the right hand side of the update equation
  \eqref{eq:th-ckt-odd-zim+i-1} of the neurons $z_{i,m+i-1}$ for odd $1<i\leq n$
  is $1$ if and only if the threshold function
  $[s_{i-1,0}'+s_{i-1,0}-2p_i - 2q_i \geq 1]$ evaluates to $1$ (this is because
  when the variables are restricted to take values in the Boolean domain
  $\inb{0,1}$, the latter happens if and only if at least one of $s_{i-1,0}'$
  and $s_{i-1,0}$ evaluates to $1$, and neither of $p_i$ and $q_i$ do).
  Similarly for the neurons $s_{i,0}$ for even $1<i<n$, we observe that the
  right hand side of the update equation \eqref{eq:th-ckt-si0-even} is $1$ if
  and only if the threshold function
  $[s_{i,0}'+z_{i+1,m+i}-2p_{i+1} - 2q_{i+1} \geq 1]$ evaluates to $1$ (this is
  because when the variables are restricted to take values in the Boolean domain
  $\inb{0,1}$, the latter happens if and only if at least one of $s_{i,0}'$ and
  $z_{i+1,m+i}$ evaluates to $1$, and neither of $p_{i+1}$ and $q_{i+1}$ do).

  We now check that $C_\phi$ has bounded degree.  Indeed, the degree of a node
  $x_{i}$, $1 \leq i \leq n$, is two more than the number of clauses in which
  the corresponding variable $x_i$ occurs, so that their degree in the graph is
  at most $6$ (since we started with a $\phi$ in which each variable occurs at
  most $4$ times).  From the description above (see also
  \cref{fig:th-ckt-construction}) we can check that all other nodes have degree
  at most $6$ (the degree of 6 is also achieved by neurons $x_{i,0}$, for odd
  $1<i\leq n$, in the counter gadget and $s_{i,0}$ ,for even $1<i<n$).
\end{proof}

We now begin the analysis of the reduction with the following two claims.
\begin{claim}
\label{obv:th-both-si0-and-zi+1m+i-is-0}
For any even $1<i<n$, $s_{i,0}(t)=z_{i+1,m+i}(t)=0$ if $t=l\cdot 2^{i} +2n +m+i$
for some integer $l \geq 1$.
\end{claim}

\begin{proof}
  From \cref{obv:th-pi-qi-update}, we see that
  $\overline{p_{i+1}(t - 1)} \land \overline{q_{i+1}(t - 1)} = 0$ when
  $t = l\cdot 2^{i} + 2n + m + i$ for some integer $l \geq 1$.  The claim then
  follows from \cref{eq:th-ckt-odd-zim+i-1,eq:th-ckt-si0-even}.
\end{proof}

\begin{claim}
\label{obv:th-si0-and-zi+1m+i-recurrence}
Fix an even $i$ satisfying $1 < i < n$ and an integer $l \geq 0$.  Suppose that
$t$ satisfying $l\cdot 2^i +2n+m+i+1\leq t \leq (l+1)\cdot 2^i +2n+m+i-1$ is
such that $s_{i,0}(t)=z_{i+1,m+i}(t)=1$.  Then we have
$s_{i,0}(t')=z_{i+1,m+i}(t')=1$ for all $t'$ satisfying
$t\leq t'\leq (l+1)\cdot 2^i +2n+m+i-1$.
\end{claim}

\begin{proof}
  From \cref{obv:th-pi-qi-update}, we see that
  $\overline{p_{i+1}}(t'') \land \overline{q_{i+1}}(t'') = 0$ only when
  $t'' = \ell'\cdot 2^{i} + 2n + m + i - 1$ for some integer $\ell' \geq 1$.  It
  follows that for all $t''$ such that
  $l\cdot 2^{i} +2n+m+i \leq t'' \leq (l+1)\cdot 2^{i} +2n+m+i-2$ for some
  integer $l\geq 0$, we have
  $\overline{p_{i+1}}(t'') \land \overline{q_{i+1}}(t'') = 1$.  Thus, when
  $l\cdot 2^{i} + 2n + m + i \leq t'' \leq (l+1) \cdot 2^{i} + 2n + m + i-2$,
  the update equations in \cref{eq:th-ckt-odd-zim+i-1,eq:th-ckt-si0-even}
  simplify to
  \begin{align*}
    s_{i,0}(t'' + 1) &= s_{i,0}'(t'') \lor z_{i+1,m+i}(t''), \text{and}\\
    z_{i+1,m+i}(t'' + 1) &= s_{i,0}'(t'') \lor s_{i,0}(t'').
  \end{align*}
  It then follows that if $t$ is such that
  $l\cdot 2^{i} +2n+m+i \leq t \leq (l+1)\cdot 2^{i} +2n+m+i-2$ and
  $s_{i,0}(t) = z_{i+1,m+i}(t) = 1$, then it is also the case that
  $s_{i,0}(t+1) = z_{i+1, m+i}(t + 1)= 1$.  The claim then follows immediately.
\end{proof}

Next we record an observation based on the properties of the counter \nckt{}
proved in \cref{cor:th-stimulation-of-x_i_(n+1-i)}.
\begin{claim}\label{obv:th-base-neuron}
  For any $t \geq 0$, $s_{1, 0}(t + m + 1) = 1$ if and only if
  $\vec{x}(t) \defeq x_1(t)x_2(t) \dots x_n(t)$ is a satisfying assignment for
  $\phi$.  In particular, for any $t'\leq 2n+m+1$, $s_{1,0}(t') = 0$.  Further,
  for any $t' \geq 2n + m + 1$, $s_{1,0}(t') = 1$ if and only if
  $t' = 2n + m + 1 + \sum_{i=1}^na_i2^{i-1} \mod 2^n$ for some satisfying
  assignment $a_1, a_2, \dots, a_n$ of $\phi$.
\end{claim}
\begin{proof}
  From \Cref{obv:th-d-s-update}, we see that the neuron $s_{1,0}$ stimulates at
  some time $t' \geq m + 1$ if and only if all the clause neurons $c_i$
  stimulate at the time $t'-m$. The latter in turn, happens if and only if the
  stimulation state of the neurons $x_1, x_2, \dots, x_n$ forms a satisfying
  assignment for $\phi$ at time $t'-m-1$.  This proves the first part of the
  claim.

  For the second part, \Cref{cor:counter-bit-values,cor:th-counter-x1-2n} imply
  that for all $1\leq i \leq n$, and for $t\leq 2n$, we have $x_i(t)=0$ (recall
  that $x_{1}$ is identical to the neuron $x_{1, 2n}$ of the counter gadget
  considered in \Cref{cor:counter-bit-values,cor:th-counter-x1-2n}, while for
  $2 \leq i \leq n$, $x_i$ is identical to the neuron $x_{i, 2n + 4 - 2i}$ of
  the same gadget). On the other hand, from our assumption on the TQBF instance,
  $\phi(x_1, x_2, \dots, x_n)$ is false when $x_1=x_2=\ldots =x_n=0$.  The first
  part proved above thus implies that for any $t'$ such that
  $m + 1 \leq t' \leq 2n+m+1$, we have $s_{1,0}(t') = 0$. When
  $0 \leq t' \leq m$, we have $s_{i, 0}(t') = 0$ from
  \Cref{obv:th-d-s-update}. These prove the second part of our claim.

  For the third part, \Cref{cor:counter-bit-values,cor:th-counter-x1-2n} imply that for
  $\vec{a} \neq \vec{0}$, $\vec{x}(t) = a_1a_2\dots a_n$ if and only if
  $t = 2n + \sum_{i=1}^{n}a_i 2^{i-1} \mod 2^n$.  This part therefore follows now
  from the already proved first part (since, by the assumption on the TQBF
  instance, $\vec{0}$ is not a satisfying assignment of $\phi$).
\end{proof}

Finally, we state the following main claim regarding the behavior of the \nckt{}
$C_\phi$.

\begin{lemma}
  \label{lem:th-induction-circuit}
  The times at which the nodes in $C_\phi$ labeled $s_{i,0}'$, $s_{i,1}$ (for even $1 < i < n$) and $s_{i,0}$ (for
  $ 1 \leq i \leq n$)  stimulate are
  characterized as follows:
  \begin{enumerate}
  \item \label{item:ckt-1} For any $t \geq 2n$, and odd $1 \leq i \leq n$, we
    have $s_{i, 0}(t  + m + i) = 1$ if and only if the formula
    \[
      \forall x_{i-1} \exists x_{i-2} \dots \forall x_2 \exists x_1 \;
      \phi(x_1, x_2, \dots, x_{i-1}, a_i, a_{i+1}, \dots, a_n)
    \]
    is true when $a_j = x_j(t)$ for all $j \geq i$.  Further, $s_{i,0}(t') = 0$
    for all $t' < 2n + m + i$.

  \item \label{item:ckt-2} For any $t \geq 2n$, and even $1 \leq i \leq n$, we have
    $s_{i, 1}(t + m + i) = 1$ if and only if $x_{i}(t) = 1$ and the formula
    \[
      \forall x_{i-2} \exists x_{i-3} \dots \forall x_2 \exists x_1 \;
      \phi(x_1, x_2, \dots, x_{i-2}, a_{i-1}, a_i, a_{i+1}, \dots, a_n)
    \]
    is true when $a_j = x_j(t)$ for all $j \geq i - 1$. Further,
    $s_{i,1}(t') = 0$ for all $t' < 2n + m + i$.

    \item \label{item:ckt-3} For any $t \geq 2n$, and even $1 \leq i \leq n$, we have
    $s_{i, 0}'(t + m + i) = 1$ if and only if $x_{i}(t) = 0$ and the formula
    \[
      \forall x_{i-2} \exists x_{i-3} \dots \forall x_2 \exists x_1 \;
      \phi(x_1, x_2, \dots, x_{i-2}, a_{i-1}, a_i, a_{i+1}, \dots, a_n)
    \]
    is true when $a_j = x_j(t)$ for all $j \geq i - 1$. Further,
    $s_{i,0}'(t') = 0$ for all $t' < 2n + m + i$.

  \item \label{item:ckt-4} For any $t \geq 2n+1$, and even $1 \leq i \leq n$, we
    have $s_{i, 0}(t + m + i) = 1$ if and only if there exists $t'$ satisfying
    $2n + \floor{(t-2n)/2^i}\cdot 2^i \leq t' < t$ such that $x_{i}(t') = 0$
    and the formula
    \[
      \forall x_{i-2} \exists x_{i-3} \dots \forall x_2 \exists x_1 \;
      \phi(x_1, x_2, \dots, x_{i-2}, a_{i-1}, a_i, a_{i+1}, \dots, a_n)
    \]
    is true when $a_j = x_j(t')$ for all $j \geq i - 1$.  Further,
    $s_{i,0}(t) = 0$ for all $t < 2n + m + i+1$.
  \end{enumerate}
\end{lemma}
\begin{proof}
  We prove the lemma by induction on the value of $i$.  In the base case,
  $i = 1$, \cref{item:ckt-1} follows immediately from \cref{obv:th-base-neuron},
  which characterizes the times $t' \geq 0$ at which $s_{1,0}(t') = 1$, and
  \cref{item:ckt-2,item:ckt-3,item:ckt-4} are vacuously true as $i$ is odd.  Thus, the base
  case is established.

  For the induction, we suppose that for some $2 \leq k \leq n$,
  \cref{item:ckt-1,item:ckt-2,item:ckt-3,item:ckt-4} of the lemma are true for all $i$
  satisfying $1 \leq i \leq k - 1$, and show that this implies that they remain
  true for $i = k$ as well.  We divide the inductive step into two cases, based
  on whether $k$ is even or odd.

  \noindent \textbf{Case 1: $k > 1$ is odd.}  In this case
  \cref{item:ckt-2,item:ckt-3,item:ckt-4} are vacuously true, as $i = k$ is odd.  Thus,
  only \cref{item:ckt-1} remains to be proved.  From \cref{eq:ckt-si0-odd}, we
  know that for any $t' \geq 1$, $s_{k,0}(t') = 1$ if and only if
  $s_{k-1,0}(t' - 1) = s_{k-1, 1}(t' - 1) = 1$. From the induction hypothesis,
  \cref{item:ckt-2} of the lemma is true for $i = k - 1$, so that we have
  $s_{k-1,1}(t' - 1) = 0$ when $1 \leq t' < 2n + m + k$ (and, also $s_{k, 0}(0) = 0$
  according to the initialization conditions). Thus, we get
  $s_{k, 0}(t') = 0$ when $ 0 \leq t' < 2n + m + k$, which establishes the second
  part of \cref{item:ckt-1}.  We now proceed to prove the first part of
  \cref{item:ckt-1}.

  For any $t \geq 2n$, we have (again from \cref{eq:ckt-si0-odd}) that
  $s_{k,0}(t + m + k) = 1$ if and only if
  $s_{k-1,0}(t + m + k - 1) = s_{k-1,1}(t + m + k - 1) = 1$.  From the induction
  hypothesis (specifically, \cref{item:ckt-2,item:ckt-4} of the lemma) applied
  with $i = k - 1$,  the latter holds if and only if the following
  two conditions are satisfied:
  \begin{enumerate}[(a)]
  \item \label{item:case-1-1} $x_{k-1}(t) = 1$ and the formula
    \[
      \forall x_{k-3} \exists x_{k-4} \dots \forall x_2 \exists x_1 \;
      \phi(x_1, x_2, \dots, x_{k-3}, a_{k-2}, a_{k-1}, a_{k}, \dots, a_n)
    \]
    is true when $a_j = x_j(t)$ for all $j \geq k - 2$.
  \item \label{item:case-1-0} there exists a $t'$ satisfying
    $2n + \floor{(t-2n)/2^{k-1}}\cdot 2^{k-1} \leq t' < t$ such that
    $x_{k-1}(t') = 0$ and the formula
    \[
      \forall x_{k-3} \exists x_{k-4} \dots \forall x_2 \exists x_1 \; \phi(x_1,
      x_2, \dots, x_{k-3}, a_{k-2}', a_{k-1}', a_{k}', \dots, a_n')
    \]
    is true when $a_j' = x_j(t')$ for all $j \geq k - 2$.
  \end{enumerate}
  Now, from \cref{cor:th-stimulation-of-x_i_(n+1-i)}, we know that for all
  $t'' \geq 2n$, the $n$-bit binary integer
  $x_n(t'')$$x_{n-1}(t'')\ldots$$x_1(t'')$ is exactly the remainder obtained on
  dividing $t'' - 2n$ by $2^n$.  Since
  $\floor{(t' - 2n)/2^{k-1}} = \floor{(t - 2n)/2^{k-1}}$, this implies that we
  have $x_j(t) = x_j(t')$ for all $j \geq k$, so that, we have $a_j = a_j'$ for
  all $j \geq k$.  Thus, we conclude that for $t \geq 2n$, we have
  $s_{k, 0}(t + m + k) = 1$ if and only if the following condition (equivalent
  to the conjunction of the conditions in \cref{item:case-1-1,item:case-1-0}
  above) holds: there exist $a_{k-2}, a_{k-2}' \in \inb{0, 1}$ such that both
  the formulas
  \[
    \forall x_{k-3} \exists x_{k-4} \dots \forall x_2 \exists x_1 \;
    \phi(x_1, x_2, \dots, x_{k-3}, a_{k-2}, 1 , a_{k}, \dots, a_n)
  \]
  and
  \[
    \forall x_{k-3} \exists x_{k-4} \dots \forall x_2 \exists x_1 \;
    \phi(x_1, x_2, \dots, x_{k-3}, a_{k-2}', 0, a_{k}, \dots, a_n)
  \]
  are true, when $a_j = x_j(t)$ for all $j \geq k$.  But this latter condition
  is equivalent to the condition that the formula
  \[
    \forall x_{k-1}\exists x_{k-2} \forall x_{k-3} \exists x_{k-4} \dots \forall
    x_2 \exists x_1 \; \phi(x_1, x_2, \dots, x_{k-3}, x_{k-2}, x_{k-1} , a_{k},
    \dots, a_n)
  \]
  is true.  This establishes the claim in \cref{item:ckt-1} of the lemma for
  $i = k$.

  \noindent \textbf{Case 2: $ 2 \leq k < n$ is even.} In this case,
  \cref{item:ckt-1} of the lemma is vacuously true, as $i = k$ is even.  Thus
  only \cref{item:ckt-2,item:ckt-3,item:ckt-4} remain to be proved.  We first consider
  \cref{item:ckt-2}.

  From \cref{eq:ckt-si1-even}, we know that for any $t' \geq 1$,
  $s_{k,1}(t') = 0$ if $s_{k-1,0}(t' - 1) = 0$ (and also that $s_{k, 1}(0) = 0$,
  which is enforced by the initial condition). From the induction
  hypothesis, \cref{item:ckt-1} of the lemma is true for $i = k - 1$, so that we
  have $s_{k-1, 0}(t' - 1) = 0$ when $1 \leq t' < 2n + m + k$.  Thus, we get
  $s_{k, 1}(t') = 0$ when $ 0 \leq t' < 2n + m + k$, which establishes the second
  part of \cref{item:ckt-2}.  We now proceed to prove the first part of
  \cref{item:ckt-2}.

  For any $t \geq 2n$, we have (again from \cref{eq:ckt-si1-even}) that
  $s_{k,1}(t + m + k) = 1$ if and only if
  $s_{k-1,0}(t + m + k - 1) = z_{k, m + k - 1}(t + m + k - 1) = 1$.  From
  \cref{obv:th-z-even-update},  $z_{k, m + k - 1}(t + m + k - 1) = 1$ if and
  only if $x_k(t) = 1$.  On the other hand, from the induction hypothesis
  (specifically, \cref{item:ckt-1} of the lemma) applied with $i = k - 1$,  $s_{k - 1, 0}(t + m + k - 1) = 1$ if and only if the formula
  \[
    \forall x_{k - 2} \exists x_{k - 3} \dots \forall x_2 \exists x_1 \;
    \phi(x_1, x_2, \dots, x_{k - 2}, a_{k - 1}, a_{k}, a_{k + 1}, \dots, a_n)
  \]
  is true when $a_j = x_j(t)$ for all $j \geq k - 1$.  Together, these two
  observations finish the proof of \cref{item:ckt-2} for $i = k$.

  We now proceed to prove \cref{item:ckt-3} for $i=k$.  We note that the
  argument is virtually identical on the one already given for
  \cref{item:ckt-2}, but we include the details for completeness.  From
  \cref{eq:th-ckt-si0'}, we know that for any $t' \geq 1$, $s_{k,0}'(t') = 0$ if
  $s_{k-1,0}(t' - 1) = 0$ (and also that $s_{k, 0}'(0) = 0$, which is enforced
  by the initial condition). From the induction hypothesis,
  \cref{item:ckt-1} of the lemma is true for $i = k - 1$, so that we have
  $s_{k-1, 0}(t' - 1) = 0$ when $1 \leq t' < 2n + m + k$.  Thus, we get
  $s_{k, 0}'(t') = 0$ when $ 0 \leq t' < 2n + m + k$, which establishes the
  second part of \cref{item:ckt-3}.  We now proceed to prove the first part of
  \cref{item:ckt-3}.

  For any $t \geq 2n$, we have (again from \cref{eq:th-ckt-si0'}) that
  $s_{k,0}'(t + m + k) = 1$ if and only if
  $s_{k-1,0}(t + m + k - 1) = 1$ and $z_{k, m + k - 1}(t + m + k - 1) = 0$.  From
  \cref{obv:th-z-even-update},  $z_{k, m + k - 1}(t + m + k - 1) = 0$ if and
  only if $x_k(t) = 0$.  On the other hand, from the induction hypothesis
  (specifically, \cref{item:ckt-1} of the lemma) applied with $i = k - 1$,  $s_{k - 1, 0}(t + m + k - 1) = 1$ if and only if the formula
  \[
    \forall x_{k - 2} \exists x_{k - 3} \dots \forall x_2 \exists x_1 \;
    \phi(x_1, x_2, \dots, x_{k - 2}, a_{k - 1}, a_{k}, a_{k + 1}, \dots, a_n)
  \]
  is true when $a_j = x_j(t)$ for all $j \geq k - 1$.  Together, these two
  observations finish the proof of \cref{item:ckt-3} for $i = k$.

  We now proceed to prove \cref{item:ckt-4} for $i = k$.  Suppose, for the sake
  of contradiction, that there exists a time $t'$, where
  $0 \leq t' < 2n + m + k+1$, such that $s_{k,0}(t') = 1$.  Without loss of
  generality, choose $t'$ to be the smallest such time.  Note that $t' > 1$, as
  $s_{k,0}(0) = 0$ is enforced by the initial conditions, while $s_{k,0}(1) = 0$
  (using \cref{eq:th-ckt-si0-even}) because both
  $s_{k,0}'(0) = z_{k+1,m+k}(0) = 0$ (as enforced by the initial condition).  On
  the other hand, we also have $s_{k, 0}(t' - 1)=s_{k, 0}(t' - 2) = 0$, using
  the choice of $t'$ as the smallest time for which $s_{k, 0}(t') = 1$, and by
  the observation above that $t' > 1$.  Further, as we have already established
  \cref{item:ckt-3} of the lemma for $i = k$ using the induction hypothesis, we
  can apply the second part of that item to deduce that
  $s_{k, 0}'(t' - 1) =s_{k, 0}'(t' - 2)= 0$ (since
  $t' - 2 < t'-1 < 2n + m + k$).  But then, \cref{eq:th-ckt-odd-zim+i-1} implies
  that $z_{k+1,m+k}(t'-1) = 0$, which further implies $s_{k,0}(t')=0$ (from
  \cref{eq:th-ckt-si0-even} as $z_{k+1,m+k}(t'-1) = s_{k, 0}'(t' - 1)=0$ ).  The
  latter is a contradiction to the assumption that $s_{k, 0}(t') = 1$.  We
  therefore must have $s_{k, 0}(t') = 0$ for all $t' < 2n + m + k+1$.  This
  establishes the second part of \cref{item:ckt-4}.  We now proceed to prove the
  first part of \cref{item:ckt-4}.  We begin by recording a few observations.

  Define $\ell_k \defeq \floor{(t-2n)/2^k} \geq 0$, and
  $t_k \defeq 2n + \floor{(t - 2n)/2^{k}}\cdot 2^{k} = 2n + \ell_k\cdot 2^k$.
  Thus, if $t'$ is as in the statement of \cref{item:ckt-4}, we have
  \begin{equation}
    \label{eq:item-4-help-1}
    t_k = \ell_k\cdot 2^k +2n  \leq t' \leq t - 1 \leq  (\ell_k+1)\cdot 2^k + 2n - 2.
  \end{equation}
  From \cref{obv:th-pi-qi-update}, we see that for any $t'' \geq 2n$, we have
  $\overline{p_{k+1}(t''  + m + k)}\land \overline{q_{k+1}(t'' + m + k)} =1$ whenever $t''$ satisfies
  \begin{equation*}
      l'\cdot 2^k +2n \leq t'' \leq (l'+1)\cdot 2^k+2n-2 \text{ for some integer $l'\geq 0$.}
  \end{equation*}
  In view of \cref{eq:item-4-help-1}, this implies that
  \begin{equation}
    \label{eq:item-4-help-3}
    \overline{p_{k+1}(t''  + m + k)}\land \overline{q_{k+1}(t'' + m + k)} =1
    \;\text{ if } \; t_k \leq t'' < t.
  \end{equation}

  We now prove the forward direction of the equivalence claimed in the first
  part of \cref{item:ckt-4} of the lemma.  Suppose therefore that there exists a
  $t'$ satisfying $2n + \floor{(t-2n)/2^k}\cdot 2^k = t_k \leq t' < t$, for some
  $t\geq 2n+1$, such that $x_{k}(t') = 0$ and the formula
  \[
    \forall x_{k-2} \exists x_{k-3} \dots \forall x_2 \exists x_1 \;
    \phi(x_1, x_2, \dots, x_{k-2}, a_{k-1}, a_k, a_{k+1}, \dots, a_n)
  \]
  is true when $a_j = x_j(t')$ for all $j \geq k - 1$.  Since we have already
  established \cref{item:ckt-3} of the lemma for $i=k$ (assuming the induction
  hypothesis), our assumption implies that $s_{k,0}'(t' + m + k ) = 1$.  The
  update expressions in \cref{eq:th-ckt-odd-zim+i-1,eq:th-ckt-si0-even}, along
  with the observation in \cref{eq:item-4-help-3}, then imply that
  \begin{equation}
    \label{eq:item-4-help-4}
    z_{k+1,m+k}(t'+m+k+1)=s_{k,0}(t'+m+k+1)=1.
  \end{equation}
  Since $t'$ satisfies \cref{eq:item-4-help-1}, we see therefore that the time
  $t' + m + k +1$ satisfies the hypothesis of
  \cref{obv:th-si0-and-zi+1m+i-recurrence}, with the integer $l$ in the claim
  set to $\ell_k$.  From \cref{obv:th-si0-and-zi+1m+i-recurrence}, we then get
  that $s_{k,0}(t'') =1$ for all $t''$ satisfying
  \begin{equation*}
    t'+m+k+1 \leq t'' \leq (\ell_k+1)\cdot 2^k +2n+m+k-1.
  \end{equation*}
  In particular, this implies that $s_{k,0}(t+m+k)=1$, since, as observed in
  \cref{eq:item-4-help-1},
  $t + m + k \leq (\ell_k + 1)\cdot2^k + 2 n + m + k - 1$.  This proves the
  forward direction of the equivalence claimed in the first part of
  \cref{item:ckt-4}.

  We now consider the other direction of the equivalence.  Suppose therefore
  that $s_{k, 0}(t + m + k) = 1$ for some $t \geq 2n+1$.  We first note that
  this implies $t \ne t_k$.  This is because when
  $t = t_k = 2n + \ell_k\cdot 2^k$, we must have $\ell_k \geq 1$ (since
  $t \geq 2n + 1$).  But then, \cref{obv:th-both-si0-and-zi+1m+i-is-0} would
  imply that $s_{k,0}(t+m+k)= s_{k,0}(\ell_k\cdot 2^k + 2n + m + k) = 0$.  Thus,
  we must have $t \neq t_k$.

  Now, let $t_1$ be the smallest integer satisfying both
  $s_{k, 0}(t_1 + m + k) = 1$ and $t_k < t_1 \leq t$ (such a $t_1$ exists as $t$
  satisfies both these conditions).  We claim that for such a $t_1$, we must
  have $s_{k,0}'(t_1+m+k-1)=1$.  We prove this by dividing the argument into two
  cases.

  \textbf{Case (a) ($t_1=t_k+1$):} In this case $s_{k,0}(t_k+m+k+1)=1$.  Note that
  if we establish that $z_{k+1,m+k}(t_k+m+k)=0$ it will therefore follow from
  \cref{eq:th-ckt-si0-even} that
  $s_{k,0}'(t_k + m + k) = s_{k,0}'(t_1 + m + k-1) = 1$.

  Now, we note that when $\ell_k \geq 1$, the claim
  $z_{k+1,m+k}(t_k+m+k) = z_{k+1, m+k}(\ell_k\cdot 2^k + 2n + m + k) =0 $ is
  directly implied by \cref{obv:th-both-si0-and-zi+1m+i-is-0}.  For the
  remaining case when $\ell_k = 0$ so that $t_k = 2n$, we observe that we have
  $s_{k,0}'(2n + m + k - 1) = 0$ (from the already established second part of
  \cref{item:ckt-3} of the induction hypothesis for the case $i = k$) and also
  $s_{k,0}(2n + m + k - 1) = 0$ (from the already established second part of
  \cref{item:ckt-4} of the induction hypothesis for the case $i = k$).  Together
  with \cref{eq:th-ckt-odd-zim+i-1}, these imply again that
  $z_{k+1,m+k}(2n+m+k) = 0$.  Thus, we always have $z_{k+1,m+k}(t_k+m+k) = 0$,
  and as observed above, this implies that if $s_{k,0}(t_k+m+k+1)=1$ then
  $s_{k,0}'(t_k + m + k) = s_{k,0}'(t_1 + m + k-1) = 1$.

  \textbf{Case (b) ($t_k + 2 \leq t_1 \leq t$):} In this case
  $s_{k,0}(t_1+m+k)=1$, and $s_{k,0}(t_1+m+k-1)=s_{k,0}(t_1+m+k-2)=0$, by the
  choice of $t_1$. Since $s_{k,0}(t_1+m+k)=1$, \cref{eq:th-ckt-si0-even} implies
  that either $s_{k,0}'(t_1+m+k-1)=1$ or $z_{k+1,m+k}(t_1+m+k-1)=1$. We want to
  show that $s_{k,0}'(t_1+m+k-1)=1$. Suppose, for the sake of contradiction,
  that $s_{k,0}'(t_1+m+k-1)=0$, so that we must have
  $z_{k+1,m+k}(t_1+m+k-1)=1$. \Cref{eq:th-ckt-odd-zim+i-1} then implies that
  $s_{k,0}'(t_1+m+k-2)=1$ as we have $s_{k,0}(t_1+m+k-2)=0$. Since
  $t_k \leq t_1 - 2 < t$, \cref{eq:item-4-help-3} implies that
  $\overline{p_{k+1}(t_1+m+k-2)} \land \overline{q_{k+1}(t_1+m+k-2)} = 1$.  In
  conjunction with \cref{eq:th-ckt-si0-even}, $s_{k,0}'(t_1+m+k-2)=1$ would then
  imply that $s_{k,0}(t_1+m+k-1)=1$.  This, however is a contradiction, since
  $s_{k,0}(t_k+m+k-1)=0$ by the choice of $t_1$. Thus we must have
  $s_{k,0}'(t_1+m+k-1)=1$.

  In all the cases, we therefore have $s_{k,0}'(t_1+m+k-1)=1$. Let $t' =
  t_1-1$. Then, $s_{k,0}'(t'+m+k)=1$, and we also have
  $2n \leq t_k \leq t' < t$. By \cref{item:ckt-3} of the lemma for $i=k$ (which
  we have already established above using the induction hypothesis), this
  implies that $x_k(t')=0$ and
  \[ \forall x_{k-2} \exists x_{k-3} \dots \forall x_2 \exists x_1 \; \phi(x_1,
    x_2, \dots, x_{k-2}, a_{k-1}, a_k, a_{k+1}, \dots, a_n)
  \]
  is true when $a_j = x_j(t')$ for all $j \geq k - 1$. Together, these imply the
  second direction of the equivalence claimed in first part of \cref{item:ckt-4}
  of the lemma, for the case $i = k$.  This completes the induction.
\end{proof}

\cref{lem:th-induction-circuit} immediately implies the following.
 \begin{lemma}
   \label{lem:th-QBF-reduces-to-NCS}
   The output neuron $O$ stimulates at some time $t \geq 0$ if and only if the
   quantified Boolean formula
   \begin{displaymath}
     \exists x_n \forall x_{n-1}\exists x_{n-2} \dots \forall x_2 \exists x_1 \phi(x_1, x_2,
     \dots, x_{n-1}, x_n)
   \end{displaymath}
   is true.
 \end{lemma}
 \begin{proof}
   Note that the output neuron $O$ stimulates at some time $t$ if and only
   $s_{n,0}$ stimulates at time $t-1$ (see \cref{eq:ckt-O}).  Recall that $n$ is
   odd.  Thus, from \cref{item:ckt-1} of \cref{lem:th-induction-circuit}
   (applied with $i = n$), $s_{n,0}$ stimulates at some time $t - 1$ if and only
   if there exists $a_n \in \{0, 1\}$ such that the quantified Boolean formula
   \begin{displaymath}
     \forall x_{n-1}\exists x_{n-2} \dots \forall x_2 \exists x_1 \phi(x_1, x_2,
     \dots, x_{n-1}, a_n)
   \end{displaymath}
   is true.  But this is equivalent to the statement that the quantified Boolean
   formula
   \[
     \exists x_{n}\forall x_{n-1}\exists x_{n-2} \dots \forall x_2 \exists x_1
     \phi(x_1, x_2, \dots, x_{n-1}, a_n)
   \]
   is true.  This completes the proof.
 \end{proof}
 Finally, we note that the above lemma proves
 \cref{thm:th-neural-ckt-simulation-is-PSPACE-hard}.
\end{proof}
We have shown that \nameref{ques:neural-circuit-simulation} is
PSPACE-hard. Finally we observe that \nameref{ques:neural-circuit-simulation}
problem is also in PSPACE.
\begin{proposition}
  \label{thm:NCS-is-in-PSPACE}
  \nameref{ques:neural-circuit-simulation} is in PSPACE.
\end{proposition}
\begin{proof}
  Suppose that we are given a neural circuit $G = (V, E)$ with input neuron $I$
  and output neuron $O$.  Starting with the initial condition in which only $I$
  is stimulated, we simulate $G$ for $2^{|V|} + 1$ units of time.  If $O$ does
  not stimulate before this time we return NO, otherwise we return YES.

  The total number of possible states of the neural circuit $G$ is $2^{|V|}$ (as
  every neuron has only two possible stimulation states). Further, given a
  stimulation state of the neural circuit at a given time, the stimulation state
  of the circuit at the next time step is completely specified by the update
  functions.  Thus, if $O$ does not become stimulated by time $2^{|V|} + 1$, the
  simulation must have entered a loop, thus ensuring that $O$ will never enter a
  stimulated state even if the simulation were continued indefinitely.  This
  shows that the above algorithm returns YES if and only if $G$ is non-trivial.

  Finally, we observe that the space needed to run the above algorithm is is
  $O(\poly{|V|})$ (corresponding to the space needed to store the states of the
  neurons and the time counter).
\end{proof}

\section{Corollaries of the hardness of \nameref{ques:neural-circuit-simulation}}
We now prove the hardness of problems 1-4 in Ramaswamy's list, quoted in the
introduction.  We begin with the following simple observation which follows
easily from \Cref{thm:th-neural-ckt-simulation-is-PSPACE-hard}.
\begin{proposition}
  For $k = 2$, \label{thm:2-size-degenerate-ckt-is-PSPACE-hard}
\nameref{ques:k-size-degenerate-circuit-decision}  is PSPACE-hard.
\end{proposition}

\begin{proof}
  We reduce the TQBF problem to
  \nameref{ques:k-size-degenerate-circuit-decision} for $k = 2$. Given a TQBF
  instance
  $\exists x_n \forall x_{n-1}\dots \exists x_1\phi(x_1, x_2, \dots, x_n)$ (as
  in \cref{def:tqbf}), we construct the same neural circuit $C_\phi$ as in the
  proof of \Cref{thm:th-neural-ckt-simulation-is-PSPACE-hard}
  (\cref{fig:th-ckt-construction}). From \Cref{lem:th-QBF-reduces-to-NCS},
  $\exists x_n \forall x_{n-1}\dots \exists x_1\phi$ is true if and only if
  $C_\phi$ is non-trivial. We now show that $C_\phi$ is non-trivial if and only
  if there is no degenerate circuit of $C_\phi$ of size $2$.
  
  From \cref{eq:ckt-O}, $O$ cannot stimulate if we silence the neuron
  $s_{n,0}$. This implies that if $C_\phi$ is non-trivial, then any degenerate
  circuit of $C_\phi$ must contain the neuron $s_{n,0}$.  Recall from
  \cref{def:degenerate-circuit} that the input neuron $I$ and the output neuron
  $O$ are both always contained in any degenerate circuit.  Thus, we see that if
  $C_\phi$ is non-trivial then any degenerate circuit of $C_\phi$ must contain
  the neurons $I$, $O$ and $s_{n,0}$, and therefore must have size at least $3$.
  On the other hand, in the case where $C_\phi$ is not non-trivial, the set
  $\inb{I, O}$ is a degenerate circuit of $C_\phi$ of size $2$ (this is because
  in the remaining circuit, the neuron $s_{n,0}$ is silenced, and hence, as
  observed above, $O$ can never stimulate).

  Thus, we see that $C_\phi$ has a degenerate circuit of size $2$ if and only if
  $C_\phi$ is not non-trivial. By construction, $C_\phi$ is non-trivial if and
  only if $\exists x_n \forall x_{n-1}\dots \exists x_1\phi$ is
  true. Altogether, this gives a polynomial-time reduction from the TQBF problem
  to \nameref{ques:k-size-degenerate-circuit-decision} for $k = 2$. This proves
  for $k = 2$, \nameref{ques:k-size-degenerate-circuit-decision} is PSPACE-hard.
\end{proof}

\begin{proposition}
\label{lem:minimal-degenerate-ckt-is-PSPACE-hard}
  \nameref{ques:minimal-degenerate-circuit-decision} is PSPACE-hard.
\end{proposition}

\begin{proof}
  We reduce the TQBF problem to
  \nameref{ques:minimal-degenerate-circuit-decision}. Given a TQBF instance
  $\exists x_n \forall x_{n-1}\dots \exists x_1\phi(x_1, x_2, \dots, x_n)$ (as
  defined in \cref{def:tqbf}), we construct the neural circuit $C_\phi$ (as in
  \cref{fig:th-ckt-construction}). We show that
  $\exists x_n \forall x_{n-1}\dots \exists x_1\phi$ is true if and only if
  every minimal degenerate circuit of $C_\phi$ of size at least $3$.
  
  As we have seen in the proof of
  \cref{thm:2-size-degenerate-ckt-is-PSPACE-hard},
  $\exists x_n \forall x_{n-1}\dots \exists x_1$ is true if and only if $C_\phi$
  does not have a degenerate circuit of size $2$. On the other hand, $C_\phi$
  always has a degenerate circuit of size greater than $2$ (namely, $C_\phi$
  itself).  Also, every degenerate circuit is of size at least $2$, as it must
  contain the neurons $I$ and $O$ (cf.~\cref{def:degenerate-circuit}). Thus, we
  see that $\exists x_n \forall x_{n-1}\dots \exists x_1\phi$ is true if and
  only if any minimal degenerate circuit of $C_\phi$ is of size at least $3$.
\end{proof}

\begin{proposition}
  \label{lem:minimum-degenerate-ckt-is-PSPACE-hard}
  \nameref{ques:minimum-degenerate-circuit-decision} is PSPACE-hard.
\end{proposition}

\begin{proof} This proof is identical to the proof of
  \cref{lem:minimal-degenerate-ckt-is-PSPACE-hard}.  Again, we reduce the TQBF
  problem to \nameref{ques:minimum-degenerate-circuit-decision}. Given a TQBF
  instance $\exists x_n \forall x_{n-1}\dots \exists x_1\phi$ (as defined in
  \cref{def:tqbf}), we construct the neural circuit $C_\phi$ (as in
  \cref{fig:th-ckt-construction}), and show that the instance is true if and
  only if a minimum degenerate circuit of $C_\phi$ of size $2$ does not exist.

  Arguing exactly as in the proof of
  \cref{lem:minimal-degenerate-ckt-is-PSPACE-hard}, we see that
  $\exists x_n \forall x_{n-1}\dots \exists x_1\phi$ is true if and only if any
  degenerate circuit (in particular any minimum degenerate circuit) of $C_\phi$
  is of size at least $3$.  This gives a polynomial time reduction from the TQBF
  problem to \nameref{ques:minimum-degenerate-circuit-decision}, and hence shows
  that \nameref{ques:minimum-degenerate-circuit-decision} is PSPACE-hard.
\end{proof}

We also record here a hardness result for an approximation version of the
problem.  For $\alpha > 1$, we say that a degenerate circuit $C$ of a given
neural circuit $G$ is an \emph{$\alpha$-approximate minimum degenerate circuit}
if the size of $C$ is at most $\alpha s$, where $s$ is the size of a minimum
degenerate circuit of $G$.

\begin{proposition}
\label{cor:finding-log-approx-degenerate-ckt-is-PSPACE-hard}
Fix any integer $c \geq 1$. It is PSPACE-hard to find a
$(c\cdot \ceil{\log N})$-approximate minimum degenerate circuit of an input
neural circuit of size $N$.
\end{proposition}

\begin{proof} We give a polynomial-time Turing reduction from TQBF to the
  $(c\cdot \ceil{\log N})$-approximate minimum degenerate circuit problem.
  Given a TQBF instance
  $\exists x_n \forall x_{n-1}\dots \exists x_1\phi(x_1, x_2, \dots, x_n)$ (as
  defined in \cref{def:tqbf}), we construct the neural circuit $C_\phi$ (as in
  \cref{fig:th-ckt-construction}).  Note that the size of $C_\phi$ is
  $N \leq p(n)$ for some fixed polynomial $p$.

  Now, we begin by finding find a $(c\cdot \ceil{\log N})$-approximate minimum
  degenerate circuit $B$ of $C_\phi$. If the size of $B$ is more than
  $2\cdot c \cdot \ceil{\log N}$, then we answer YES (i.e., that
  $\exists x_n \forall x_{n-1}\dots \exists x_1\phi$ is true).  This is correct
  since in this case, the size of any minimum degenerate circuit of $C_\phi$
  must be at least $3$, so that, as argued in the proof of
  \cref{lem:minimal-degenerate-ckt-is-PSPACE-hard},
  $\exists x_n \forall x_{n-1}\dots \exists x_1\phi$ must be true.

  If the size of $B$ is at most $2\cdot c \cdot \ceil{\log N}$, we simulate $B$
  for time $2^{|B|} + 1 \leq 2^{2\cdot c \cdot \log N} + 1 \leq \poly{n}$
  starting from the initial condition in which only the input neuron is
  stimulated and answer YES if the output neuron stimulates during the
  simulation and NO otherwise.  The correctness is then guaranteed by the facts
  that $B$ is a degenerate circuit of $C_\phi$, and that $C_\phi$ is non-trivial
  if and only if $\exists x_n \forall x_{n-1}\dots \exists x_1\phi$ is true.
  The sufficiency of the time of simulation follows as in the proof of
  \Cref{thm:NCS-is-in-PSPACE}.
\end{proof}

\begin{proposition}
\label{thm:1-Vital-set-is-PSPACE-hard}
\nameref{ques:1-vital-set-decision} is PSPACE-hard.\end{proposition}
\begin{proof} We reduce the TQBF problem to
  \nameref{ques:1-vital-set-decision}. Given a TQBF instance $\phi$ (as defined
  in \cref{def:tqbf}), we construct the neural circuit $C_\phi$ (as in
  \cref{fig:th-ckt-construction}). We now show that
  $\exists x_n \forall x_{n-1}\dots \exists x_1\phi$ is true if and only if the
  set of $1$-vital sets of $C_\phi$ is non-empty.

  The set of $1$-vital sets contains those neurons (except $I$ and $O$) which
  belong to all the degenerate circuits. As we have seen in the proof of
  \cref{thm:2-size-degenerate-ckt-is-PSPACE-hard},
  $\exists x_n \forall x_{n-1}\dots \exists x_1\phi$ is true if and only if all
  the degenerate circuits of $C_\phi$ contain the neurons $I$, $O$, and
  $s_{n,0}$. On the other hand, if
  $\exists x_n \forall x_{n-1}\dots \exists x_1\phi$ is false, then $C_\phi$ is
  not non-trivial so that $\inb{I, O}$ is a degenerate circuit for $C_\phi$.
  Thus, we see that $\exists x_n \forall x_{n-1}\dots \exists x_1\phi$ is true
  if and only if the set of $1$-vital sets of $C_\phi$ is non-empty. This gives
  a polynomial-time reduction from the TQBF problem to
  \nameref{ques:1-vital-set-decision}, and hence shows that
  \nameref{ques:1-vital-set-decision} is PSPACE-hard.
\end{proof}

%


 \paragraph{Acknowledgments.}
 We thank Shubham Pawar for several helpful discussions.  PS acknowledges support
 from a Ramanujan Fellowship of the DST.  VSS and PS acknowledge support from the
 Department of Atomic Energy, Government of India, under project
 no. 12-R\&D-TFR-5.01-0500.

\bibliographystyle{acm}
\bibliography{biblio}

\end{document}